\newcommand{\FULLPAPER}{}

\documentclass[letterpaper,10pt,twocolumn]{article}
\usepackage{filecontents}
\usepackage{graphicx}
\usepackage{hhline}
\usepackage{xspace}
\usepackage{balance}
\usepackage{ulem}
\usepackage{url}
\usepackage{titlesec}
\usepackage{colortbl}
\usepackage{multirow}
 \usepackage{tikz}

\usepackage{amssymb}
\usepackage{amsmath}
\usepackage{enumerate}
\usepackage{enumitem}
\usepackage[color=yellow,prependcaption,textsize=normal]{todonotes}

\usepackage{algorithmicx}
\usepackage{algorithm}
\usepackage{algpseudocode}
\usepackage{varwidth}

\usepackage{amsthm}
\usepackage{gensymb}
\usepackage{makecell}
\usepackage{tabularx, booktabs}
\usepackage{array}
\usepackage{times}
\usepackage{fullpage}
\usepackage{dsfont}
\usepackage{epsfig}
\usepackage{microtype}
\usepackage{subcaption}

\newtheorem{definition}{Definition}

\newcommand{\omitit}[1]{}

\newtheorem{lemma}{Lemma}

\usepackage{mathtools}
\DeclarePairedDelimiter\ceil{\lceil}{\rceil}

\algnewcommand{\LineComment}[1]{\Statex \hskip\ALG@thistlm \(\triangleright\) #1}

\newcommand{\protocol}{Mir\xspace}
\newcommand{\protocolfull}{Mir-BFT\xspace}
\newcommand{\eplead}{\ensuremath{EL}\xspace}

\newcommand{\stable}{stable\xspace}
\newcommand{\ephemeral}{ephemeral\xspace}

\newcommand{\bcast}{\textsc{bcast}}
\newcommand{\deliver}{\textsc{commit}}

\newcommand{\prepreparemsg}{\textsc{PRE-PREPARE}\xspace}
\newcommand{\preparemsg}{\textsc{PREPARE}\xspace}
\newcommand{\commitmsg}{\textsc{COMMIT}\xspace}
\newcommand{\checkpointmsg}{\textsc{CHECKPOINT}\xspace}
\newcommand{\requestmsg}{\textsc{REQUEST}\xspace}

\newcommand{\faults}{f}

\newcommand{\event}[3]{
    \ifthenelse
    {\equal{#3}{}}
    {\ap{#1.\textrm{#2}}}
    {\ap{#1.\textrm{#2} \mid #3}}
}

\algnewcommand\Instance[2]{\State #1, \textbf{instance} #2}
\algnewcommand\InstanceSystem[3]{\State #1, \textbf{instance} #2, \textbf{system} #3}
\algnewcommand\Trigger[3]{\State \textbf{trigger} $\event{#1}{#2}{#3}$}
\algnewcommand\Schedule[1]{\State \textbf{schedule} $#1$}
\algnewcommand\Cancel[1]{\State \textbf{cancel} $#1$}
\algnewcommand\Broadcast[1]{\State \textbf{broadcast} $#1$}
\algnewcommand\Import[1]{\State \textbf{import} $#1$}

\algnewcommand\Not{\textbf{ not }}
\algnewcommand\AndT{\textbf{ and }}
\algnewcommand\OrT{\textbf{ or }}
\algnewcommand\In{\textbf{ in }}

\algsetblockdefx[Implements]{Implements}{EndImplements}{}{}{\textbf{Implements:}}{}
\algsetblockdefx[Uses]{Uses}{EndUses}{}{}{\textbf{Uses:}}{}
\algsetblockdefx[Parameters]{Parameters}{EndParameters}{}{}{\textbf{Parameters:}}{}
\algsetblockdefx[Init]{Init}{EndInit}{}{}{\textbf{Init:}}{}

\algsetblockdefx[Upon]{Upon}{EndUpon}{}{}[1]{\textbf{upon} $#1$ \textbf{do}}{}
\algsetblockdefx[UponExists]{UponExists}{EndUponExists}{}{}[2]{\textbf{upon exists} $#1$ \textbf{such that} $#2$ \textbf{do}}{}
\algsetblockdefx[UponCondition]{UponCondition}{EndUponCondition}{}{}[1]{\textbf{upon} $#1$ \textbf{do}}{}
\algsetblockdefx[UponReceiving]{UponReceiving}{EndUponReceiving}{}{}[1]{\textbf{upon receiving} $#1$ \textbf{do}}{}

\algsetblockdefx[Procedure]{Procedure}{EndProcedure}{}{}[2]{\textbf{procedure} \emph{#1}($#2$) \textbf{is}}{}
\algsetblockdefx[Function]{Function}{EndFunction}{}{}[2]{\textbf{function} \emph{#1}($#2$) \textbf{:}}{}
\algsetblockdefx[FunctionNoArgs]{FunctionNoArgs}{EndFunctionNoArgs}{}{}[1]{\textbf{function} \emph{#1}() \textbf{:}}{}

\algdef{SE}[DoUntil]{DoUntil}{EndDoUntil}{\textbf{do}}[1]{\textbf{until} $#1$}
\algsetblockdefx[Struct]{Struct}{EndStruct}{}{}[1]{\textbf{Struct} $#1$ \textbf{contains}}{}
\algsetblockdefx[IfExists]{IfExists}{EndIfExists}{}{}[2]{\textbf{if exists} $#1$ \textbf{such that} $#2$ \textbf{then}}{\textbf{end if}}
\algsetblockdefx[While]{While}{EndWhile}{}{}[1]{\textbf{while} $#1$ \textbf{do}}{\textbf{end while}}
\algsetblockdefx[NTimes]{NTimes}{EndNTimes}{}{}[1]{\textbf{for} $#1$ \textbf{times do}}{\textbf{end for}}
\algsetblockdefx[For]{For}{EndFor}{}{}[3]{\textbf{for} $#1 \in #2..#3$ \textbf{do}}{\textbf{end for}}
\algsetblockdefx[ForAll]{ForAll}{EndForAll}{}{}[2]{\textbf{for all} $#1 \in #2$ \textbf{do}}{\textbf{end for}}

\newcommand{\remove}[1]{}
\newcommand{\add}[1]{#1}
\newcommand{\needsrev}[1]{}

\graphicspath{{figures/}}

\begin{document}

 \title{\textbf{\protocolfull: High-Throughput Robust BFT for Decentralized Networks}}
 \author{
    Chrysoula Stathakopoulou \\IBM Research - Zurich \and
    Tudor David \\ Oracle Labs\thanks{Work done in IBM Research - Zurich} \and
    Matej Pavlovic \\ IBM Research - Zurich \and
    Marko Vukoli\'c \\IBM Research - Zurich
    \date{}
 }

\maketitle

\begin{abstract}
  This paper presents Mir-BFT, a robust Byzantine fault-tolerant (BFT) total order broadcast protocol aimed at maximizing throughput on wide-area networks (WANs), targeting deployments in decentralized networks, such as permissioned and Proof-of-Stake permissionless blockchain systems.
  Mir-BFT is the first BFT protocol that allows multiple leaders to propose request batches independently (i.e., parallel leaders), in a way that precludes request duplication attacks by malicious (Byzantine) clients, by rotating the assignment of a partitioned request hash space to leaders.
  As this mechanism removes a single-leader bandwidth bottleneck and exposes a computation bottleneck related to authenticating clients even on a WAN, our protocol further boosts throughput using a client signature verification sharding optimization.
  Our evaluation shows that Mir-BFT outperforms state-of-the-art and orders more than 60000 signed Bitcoin-sized (500-byte) transactions per second on a widely distributed 100 nodes, 1 Gbps WAN setup, with typical latencies of few seconds.
  We also evaluate Mir-BFT under different crash and Byzantine faults, demonstrating its performance robustness.
  Mir-BFT relies on classical BFT protocol constructs, which simplifies reasoning about its correctness.
  Specifically, Mir-BFT is a generalization of the celebrated and scrutinized PBFT protocol. In a nutshell, Mir-BFT follows PBFT ``safety-wise'', with changes needed to accommodate novel features restricted to PBFT liveness.
\end{abstract}


\section{Introduction}

\noindent\textbf{Background.}
Byzantine fault-tolerant (BFT) protocols, which tolerate malicious (Byzantine \cite{Lamport:1982:BGP:357172.357176}) behavior of a subset of nodes,
have evolved from being a niche technology for tolerating bugs and intrusions to be the key technology to ensure consistency of widely deployed decentralized networks in which multiple mutually untrusted parties administer different nodes (such as in blockchain systems)
\cite{Vukolic15,CromanDEGJKMSSS16,Algorand}.
Specifically, BFT protocols are considered to be an alternative (or complementing) to  energy-intensive and slow Proof-of-Work (PoW) consensus protocols used in early blockchains including Bitcoin \cite{Vukolic15,GervaisKWGRC16}.
BFT protocols relevant to decentralized networks are consensus and total order (TO) broadcast protocols \cite{CachinGR11} which establish the basis for state-machine replication (SMR) \cite{Schneider90} and smart-contract execution \cite{EthereumYellowPaper}.

BFT protocols are known to be very efficient on small scales (few nodes) in clusters (e.g., \cite{700, Zyzzyva}), or to exhibit modest performance on large scales (thousands or more nodes) across wide area networks (WAN) (e.g., \cite{Algorand}).
Recently, considerable research effort (e.g., \cite{Buchman16,RedBelly,SBFT,hotstuff,MillerXCSS16}) focused on maximizing BFT performance in medium-sized WAN networks (e.g., at the order of 100 nodes) as this deployment setting is highly relevant to different types of decentralized networks.

On the one hand, \emph{permissioned} blockchains, such as Hyperledger Fabric \cite{AndroulakiBBCCC18}, are rarely deployed on scales above 100 nodes, yet use cases gathering dozens of organizations, which do not necessarily trust each other, are very prominent \cite{CordaURL}.
On the other hand, this setting is also highly relevant in the context of large scale \emph{permissionless} blockchains, in which anyone can participate,
that use weighted voting (based e.g., on Proof-of-Stake (PoS) \cite{CasperFFG,Ouroboros} or  delegated PoS (DPoS) \cite{TendermintURL}), or committee-voting \cite{Algorand}, to limit the number of nodes involved in the critical path of the consensus protocol.
With such weighted voting, the number of (relevant) nodes for PoS/DPoS consensus is typically in the order of a hundred \cite{TendermintURL}, or sometimes even less  \cite{EOSBlockProducers}. Related open-membership blockchain systems, such as Stellar, also run consensus among less than 100 nodes \cite{stellarsosp}.

\noindent\textbf{Challenges.} Most of this research (e.g., \cite{Buchman16,RedBelly,SBFT,hotstuff}) aims at addressing the scalability issues that arise in classical leader-based BFT protocols, such as the seminal PBFT protocol \cite{Castro:2002:PBF}.
In short, in a leader-based protocol, a leader, who is tasked with assembling a batch of requests (block of transactions) and communicating it to all other nodes, has at least $O(n)$ work, where $n$ is the total number of nodes.
Hence the leader quickly becomes a bottleneck as $n$ grows.

A promising approach to addressing scalability issues in BFT is to allow multiple nodes to act as \emph{parallel leaders} and to propose batches independently and concurrently, either in a coordinated, deterministic fashion \cite{RedBelly,BFT-Mencius}, or using randomized protocols \cite{BEAT,MillerXCSS16,Hashgraph}.
With parallel leaders, the CPU and bandwidth load related to proposing batches are distributed more evenly.
However, the issue with this approach is that parallel leaders are prone to wasting resources by proposing the same duplicate requests.
As depicted in Table~\ref{table:comparison}, none of the current BFT protocols that allow for parallel leaders deal with request duplication, which is straightforward to satisfy in single leader protocols.
The tension between preventing request duplication and using parallel leaders stems from two important attacks that an adversary can mount and an efficient BFT protocol needs to prevent: (i) the \emph{request censoring attack} by Byzantine leader(s), in which a malicious leader simply drops or delays a client's request (transaction), and (ii)  the \emph{request duplication attack}, in which Byzantine clients submit the exact same request multiple times.

\begin{table}[ht]
	\centering
	\scriptsize{
		\begin{tabular}{| l | l | l|}
			\hline
			\textbf{}                   &    \textbf{Parallel Leaders}  & \textbf{Prevents Req. Duplication}  \\ \hline
			PBFT \cite{Castro:2002:PBF}          &  no           & \textbf{yes}          \\
			BFT-SMaRt \cite{BessaniSA14}        &  no           & \textbf{yes}     \\
			Aardvark \cite{Aardvark}             &  no           & \textbf{yes}             \\
			RBFT \cite{RBFT}                   &  no           & \textbf{yes}                    \\
			Spinning \cite{Spinning}      &  no           & \textbf{yes}         \\
			Prime \cite{Prime}                 &  no           & \textbf{yes}         \\
			700 \cite{700}                       &  no           & \textbf{yes}          \\
			Zyzzyva \cite{Zyzzyva}               &  no           & \textbf{yes}         \\
			SBFT  \cite{SBFT}                     &  no           & \textbf{yes}          \\
			HotStuff \cite{hotstuff}           &  no           & (yes)          \\
			Tendermint \cite{Buchman16}         &  no           & \textbf{yes}       \\
			BFT-Mencius \cite{BFT-Mencius}     &  \textbf{yes}         & no            \\
			RedBelly \cite{RedBelly}          &  \textbf{yes}         & no          \\
			Hashgraph \cite{Hashgraph}        &  \textbf{yes}         & no              \\
			Honeybadger \cite{MillerXCSS16}       & \textbf{yes}         & no           \\
			BEAT \cite{BEAT}                    &  \textbf{yes}         & no     \\
			\textbf{\protocol (this paper)} & \textbf{yes} & \textbf{yes} \\ \hline
		\end{tabular}
	}
	\caption{Comparison of \protocol to related BFT protocols.
		By \textit{(yes)} we denote a property which can easily be satisfied.
	}
	\label{table:comparison}
\end{table}

To counteract request censoring attacks, a BFT protocol needs to allow at least $\faults+1$ different leaders to propose a request (where $\faults$, which is typically $O(n)$, is the threshold on the number of Byzantine nodes in the system).
Single-leader protocols (e.g., \cite{Castro:2002:PBF,hotstuff}), which typically rotate the leadership role across all nodes, address censoring attacks relatively easily.
On changing the leader, a new leader only needs to make sure they do not repeat requests previously proposed by previous leaders.

With parallel leaders, the picture changes substantially.
If a (malicious or correct) client submits the same request to multiple parallel leaders concurrently, parallel leaders will duplicate that request.
While these duplicates can simply be filtered out after order (or after the reception of a duplicate, during ordering), the damage has already been done --- excessive resources, bandwidth and possibly CPU have been consumed.
To complicate the picture, na\"ive solutions in which: (i) clients are requested to sequentially send to one leader at the time, or (ii) simply to pay transaction fees for each duplicate, do not help.
In the first case,  Byzantine clients mounting a request duplication attack are not required to respect sending a request sequentially and, what is more, such a behavior cannot be distinguished from a correct client who simply sends a transaction multiple times due to asynchrony or network issues.
Even though some blockchain systems, such as Hedera Hashgraph, charge transaction fees for every duplicate \cite{BairdPcomm}, this approach penalizes correct clients when they resubmit a transaction to counteract possible censoring attacks, or a slow network.
In more established decentralized systems, such as Bitcoin and Ethereum, it is standard to charge for the same transaction only once, even if it is submitted by a client more than once.
In summary, with up to $O(n)$ parallel leaders, request duplication attacks may induce an $O(n)$-fold duplication of every single request and bring the effective throughput to its knees, practically voiding the benefits of using multiple leaders.

\noindent\textbf{Contributions.} This paper presents \protocolfull (or, simply, \protocol\footnote{In a number of Slavic languages, the word \emph{mir} refers to  universally good, global concepts, such as \emph{peace} and/or \emph{world}.}),
a novel BFT total order broadcast (TOB) protocol that is the first to
combine parallel leaders with robustness to attacks \cite{Aardvark}.
In particular, \protocol precludes request duplication performance attacks, and addresses other notable performance attacks \cite{Aardvark},
such as the Byzantine leader straggler attack.
\protocol is further robust to arbitrarily long (yet finite) periods of asynchrony and is optimally resilient (requiring optimal $n\ge 3f+1$ nodes to tolerate $f$ Byzantine faulty ones).
On the performance side, \protocol achieves the best throughput to date on public WAN networks,
as confirmed by our measurements on up to 100 nodes.
The following summarizes the main features of \protocol, as well as contributions of this paper:

\newcommand{\pointone}{
    \protocol allows multiple parallel leaders to propose batches of requests concurrently,
    in a sense multiplexing several PBFT instances into a single total order, in a robust way.
    As its main novelty, \protocol partitions the request hash space across replicas,
    preventing request duplication, while rotating this partitioned assignment across protocol configurations/epochs,
    addressing the request censoring attack. Mir further uses a \emph{client signature verification sharding} throughput optimization to offload CPU,
    which is exposed as a bottleneck in \protocol once we remove the single-leader bandwidth bottleneck using parallel leaders.
}

\newcommand{\pointtwo}{
	\protocol avoids ``design-from-scratch'', which is known to be error-prone for BFT \cite{700,ZyzzyvaBug}.
	\protocol is a generalization of the well-scrutinized PBFT protocol
	\footnote{\protocol variants based on other BFT protocols can be derived as well.},
	which \protocol closely follows ``safety-wise''
	while introducing important generalizations only affecting PBFT liveness (e.g., (multiple) leader election).
	This simplifies the reasoning about \protocol  correctness.}

\newcommand{\pointfour}{
    We implement \protocol in Go and run it with up to 100 nodes in a multi-datacenter WAN,
    as well as in clusters and under different faults,
    comparing it to state of the art BFT protocols.
    Our results show that \protocol convincingly outperforms state of the art, 
    ordering more than 60000 signed Bitcoin-sized (500-byte) requests per second (req/s) on a scale of 100 nodes on a WAN,
    with typical latencies of few seconds.     In this setup, \protocol achieves 3x the throughput of the optimistic Chain of \cite{700}, and more than an order of magnitude higher throughput than other state of the art single-leader BFT protocols. To put this into perspective, \protocol's 60000+ req/s on 100 nodes on WAN is 2.5x the alleged peak capacity of VISA (24k req/s \cite{CapacityTop20URL})
    and more than 30x faster than the actual average VISA transaction rate (about 2k req/s \cite{Vukolic15}).
}

\noindent\textbf{$\bullet$}  \pointone
\noindent\textbf{$\newline\bullet$}  \pointtwo
\noindent\textbf{$\newline\bullet$}  \pointfour

\noindent\textbf{Roadmap.} The rest of the paper is organized as follows.
In Section~\ref{sec:model}, we define the system model
and in Section \ref{ssec:PBFT} we briefly present PBFT (for completeness).
In Section~\ref{sec:overview}, we give an overview of \protocol and changes it introduces to PBFT.
We then explain \protocol implementation details in Section~\ref{sec:details}.
We further list the \protocol pseudocode in Section~\ref{sec:pseudocode}.
This is followed by \protocol correctness proof in Section~\ref{sec:correctness}.
Section~\ref{sec:LTO} introduces an optimization tailored to large requests.
Section~\ref{sec:evaluation} gives evaluation details.
Finally, Section~\ref{sec:relwork} discusses related work and Section~\ref{sec:conclusions} concludes.

\begin{figure*}[htbp]
    \begin{center}
        \includegraphics[width=\textwidth]{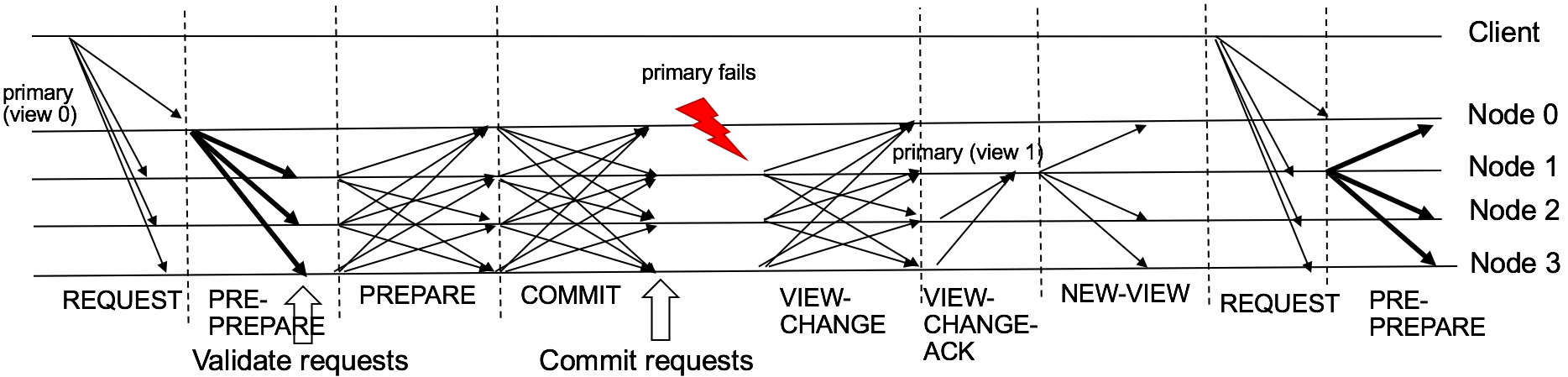}
        \caption{PBFT communication pattern and messages. Bottleneck messages are shown in \textbf{bold}.}
        \label{fig:pbft}
    \end{center}
\end{figure*}


\section{System Model}
\label{sec:model}

We assume an eventually synchronous system \cite{Dwork:1988:CPP:42282.42283}
in which the communication among \emph{correct} processes can be fully asynchronous before some time $GST$, unknown to nodes,
after which it is assumed to be synchronous.
Processes are split into a set of $n$ \emph{nodes} (the set of all nodes is denoted by $Nodes$) and a set of \emph{clients}.
We assume a public key infrastructure in which processes are identified by their public keys;
we further assume node identities are lexicographically ordered and mapped by a bijection to the set $[0\ldots n-1]$
which we use to reason about node identities.
In every execution, at most $\faults$ nodes can be \emph{Byzantine} faulty
(i.e., crash or deviate from the protocol in an arbitrary way),
such that $n\ge 3\faults+1$.
Any number of clients can be Byzantine.

We assume an adversary that can control Byzantine faulty nodes but cannot break the cryptographic primitives we use,
such as PKI and cryptographic hashes (we use SHA-256).
$H(data)$ denotes a cryptographic hash of $data$,
while $data_{\sigma_p}$ denotes $data$ signed by process $p$ (client or node).
Processes communicate through authenticated point-to-point channels
(our implementation uses gRPC \cite{gRPCURL} over TLS, preventing man-in-the-middle and related attacks).

Nodes implement a BFT total order (atomic) broadcast service to clients.
To broadcast request $r$, a client invokes $\bcast(r)$,
with nodes eventually outputting  $\deliver(sn,r)$, such that the following properties hold:

\begin{enumerate}[label=P\arabic*]
    \item \textbf{Validity:} If a correct node commits $r$ then some client broadcasted $r$.
    \item \textbf{Agreement (Total Order):} If two correct nodes commit requests $r$ and $r'$ with sequence number $sn$, then $r=r'$.
    \item \textbf{No duplication:} If a correct node commits request $r$ with sequence numbers $sn$ and  $sn'$, then  $sn = sn'$.
    \item \textbf{Totality:} If a correct node commits request $r$, then every correct node eventually commits $r$.
    \item \textbf{Liveness:} If a correct client broadcasts request $r$, then some correct node $p$ eventually commits $r$.
\end{enumerate}

Note that P3 (No duplication) is a standard TOB property \cite{CachinGR11} that most protocols can easily satisfy by filtering out duplicates \emph{after} agreeing on request order, which is bandwidth wasting.
\protocol enforces P3 \emph{without ordering duplicates}, using a novel approach to eliminate duplicates during agreement to improve performance and scalability.


\section{PBFT and its Bottlenecks}
\label{ssec:PBFT}

\begin{table*}[ht]
    \centering
    \footnotesize{
        \begin{tabular}{| l | l | l | }
            \hline
            \textbf{Protocol} & \textbf{PBFT} \cite{Castro:2002:PBF} & \textbf{\protocol}  \\ \hline
            Client request authentication & vector of MACs (1 for each node)  & signatures \\ \hline
            Batching & no (or, 1 request per ``batch'') & yes \\ \hline
            Multiple-batches in parallel & yes (watermarks) & yes (watermarks) \\ \hline
            Round structure/naming & views & epochs \\ \hline
            Round-change responsibility & view primary (round-robin across all nodes) & epoch primary (round-robin across all nodes)  \\ \hline
            No. of per-round leaders & 1 (view primary) & many (from $1$ to $n$ epoch leaders) \\ \hline
            No. of batches per round & unbounded & bounded (\emph{ephemeral} epochs); unbounded (\emph{\stable} epochs) \\ \hline
            Round leader selection & primary is the only leader & primary decides on epoch leaders (subject to constraints) \\ \hline
            Request duplication prevention & enforced by the primary & hash space partitioning across epoch leaders (rotating) \\ \hline
        \end{tabular}
    }
    \caption{High level overview of the original PBFT \cite{Castro:2002:PBF} vs. \protocol protocol structure.}
    \label{table:diff}
\end{table*}

We depict the PBFT communication pattern in Figure~\ref{fig:pbft}.
PBFT proceeds in rounds called \emph{views} which are led by the \emph{primary}.
The primary sequences and \emph{proposes} a client's request (or a batch thereof) in a \prepreparemsg message
--- on WANs this step is typically a network bottleneck.
Upon reception of the \prepreparemsg, other nodes validate the request,
which involves, at least, verifying its authenticity
(we say nodes \emph{pre-prepare} the request).
This is followed by two rounds of all-to-all communication (\preparemsg and \commitmsg messages),
which are not bottlenecks as they leverage $n$ links in parallel and contain metadata (request/batch hash) only.
A node \emph{prepares} a request and sends a \commitmsg message
if it gets a \preparemsg message from a quorum ($n-f\ge 2\faults+1$ nodes) that matches a \prepreparemsg.
Finally, nodes \emph{commit} the request in total order, if they get a quorum of matching \commitmsg messages.

The primary is changed only if it is faulty or if asynchrony breaks the availability of a quorum.
In this case, nodes timeout and initiate a \emph{view change}.
View change involves communication among nodes
in which they exchange information about the latest \emph{pre-prepared} and \emph{prepared} requests,
such that the new primary, which is selected in round robin fashion,
must re-propose potentially committed requests under the same sequence numbers
within a \textsc{new-view} message (see \cite{Castro:2002:PBF} for details).
The view-change pattern can be simplified using signatures \cite{Castro:1999:PBF:296806.296824}.

After the primary is changed,
the system enters the new view and common-case operation resumes.
PBFT complements this main common-case/view-change protocols with \emph{checkpointing} (log and state compaction)
and state transfer subprotocols \cite{Castro:2002:PBF}.


\section{\protocol Overview}
\label{sec:overview}

\protocol is based on PBFT \cite{Castro:2002:PBF} (Sec.~\ref{ssec:PBFT})
--- major differences are summarized in Table~\ref{table:diff}.
In this section we elaborate on these differences, giving a high-level overview of \protocol.

\paragraph{Request Authentication.}
While PBFT authenticates clients' requests with a vector of MACs, \protocol uses signatures for request authentication
to avoid concerns associated with ``faulty client'' attacks
related to the MAC authenticators PBFT uses \cite{Aardvark} and to prevent any number of colluding nodes, beyond $f$, from impersonating a client.
However, this change may induce a throughput bottleneck,
as per-request verification of clients' signatures requires more CPU than that of MACs.
We address this issue by a signature verification sharding optimization described in Sec.~\ref{sec:sharding}.

\paragraph{Batching and Watermarks.}
\protocol processes requests in \emph{batches} (ordered lists of requests formed 
by a leader),
a standard throughput improvement of PBFT (see e.g., \cite{Zyzzyva,700}).
\protocol also retains request/batch \emph{watermarks} used by PBFT to boost throughput.
In PBFT, request watermarks, low and high, represent the range of request sequence numbers
which the primary/leader can propose concurrently.
While many successor BFT protocols eliminated watermarks in favor of batching (e.g, \cite{Zyzzyva,BessaniSA14,700}),
\protocol reuses watermarks to facilitate concurrent proposals of batches by \emph{multiple parallel leaders}.

\paragraph{Protocol Round Structure.}
Unlike PBFT, \protocol distinguishes between leaders and a primary node.
\protocol proceeds in \emph{epochs} which correspond to \emph{views} in PBFT, each epoch having a single \emph{epoch primary} --- a node deterministically defined by the epoch number,
by round-robin rotation across all the participating nodes of the protocol.

Each epoch $e$ has a set of \emph{epoch leaders} (denoted by $\eplead(e)$),
which we define as nodes that can sequence and propose batches in $e$
(in contrast, in PBFT, only the primary is a leader).
Within an epoch, \protocol deterministically partitions sequence numbers across epoch leaders,
such that all leaders can propose their batches simultaneously without conflicts. Epoch $e$ transitions to epoch $e+1$ if
(1) one of the leaders is suspected of failing, triggering a timeout at sufficiently many nodes (\emph{ungracious epoch change}), or
(2) a predefined number of batches $maxLen(e)$ has been committed (\emph{gracious epoch change}).
While the ungracious epoch change corresponds exactly to PBFT's view change,
the gracious epoch change is a much more lightweight protocol.

\paragraph{Selecting Epoch Leaders.}
For each epoch, it is the \emph{primary} who \emph{selects the leaders} and reliably broadcasts its selection to all nodes.
In principle, the primary can pick an arbitrary leader set as long as the primary itself is included in it.
We evaluated a simple ``grow on gracious, reduce on ungracious epoch'' policy for leader set size.
If primary $i$ starts epoch $e$ with a gracious epoch change it adds itself to the leader set of the preceding epoch $e-1$. If $i$  starts  epoch $e$ with an ungracious epoch change and $e'$ is the last epoch for which $i$ knows the epoch configuration, $i$ adds itself to  the leader set  of epoch $e'$ and removes one node (not itself) for each epoch between $e$ and $e'$ (leaving at least itself in the leader set).%
More elaborate leader set choices, which are outside the scope of this paper, can take into account different heuristics, execution history, fault patterns, weighted voting, distributed randomness, or blockchain stake.

Moreover, in an epoch $e$ where all nodes are leaders ($\eplead(e) = Nodes$),
we define $maxLen(e) = \infty$ (i.e., $e$ only ends if a leader is suspected).
Otherwise, $maxlen(e)$ is a constant, pre-configured system parameter.
We call the former \emph{\stable} epochs and the latter \emph{\ephemeral}.

More elaborate strategies for choosing epoch lengths and leader sets,
which are outside the scope of this paper,
can take into account execution history, fault patterns, weighted voting, distributed randomness, or blockchain stake.
Note that with a policy that constrains the leader set to only the epoch primary and makes every epoch \stable,
\protocol reduces to PBFT.

\paragraph{Request Duplication and Request Censoring Attacks.}
Moving from single-leader PBFT to multi-leader \protocol poses the challenge of request duplication.
A simplistic approach to multiple leaders would be to allow any leader to add any request into a batch
(\cite{BFT-Mencius, RedBelly, Hashgraph}),
either in the common case, or in the case of client request retransmission.
Such a simplistic approach, combined with a client sending a request to exactly one node,
allows good throughput with no duplication only in the best case,
i.e., with no Byzantine clients/leaders and with no asynchrony.

However, this approach does not perform well outside the best case,
in particular with clients sending identical request to multiple nodes. A client may do so simply because it is Byzantine and performs the \emph{request duplication} attack.
However, even a correct client needs to send its request to at least $\faults+1$ nodes
    (i.e., to $\Theta(n)$ nodes, when $n=3f+1$) in the worst case in \emph{any} BFT protocol,
    in order to avoid Byzantine nodes (leaders) selectively ignoring the request (\emph{request censoring} attack). Therefore, a simplistic approach to parallel request processing with multiple leaders \cite{BFT-Mencius, RedBelly, Hashgraph}
faces attacks that can reduce throughput by factor of $\Theta(n)$,
nullifying the effects of using multiple leaders.

Note the subtle but important difference between a duplication attack (submitting the same request to multiple replicas)
and a DoS attack (submitting many different requests) that a Byzantine client can mount.
A system can prevent the latter (DoS) by imposing per-client limits on incoming \emph{unique} request rate.
Mir enforces such a limit through client request watermarks.
A duplication attack, however, is resistant to such mechanisms, as a Byzantine client
is indistinguishable from a correct client with a less reliable network connection.
We demonstrate the effects of these attacks in Section~\ref{sec:evalfaults}.

\begin{figure}[htbp]
	\begin{center}
		\includegraphics[width=0.79\columnwidth]{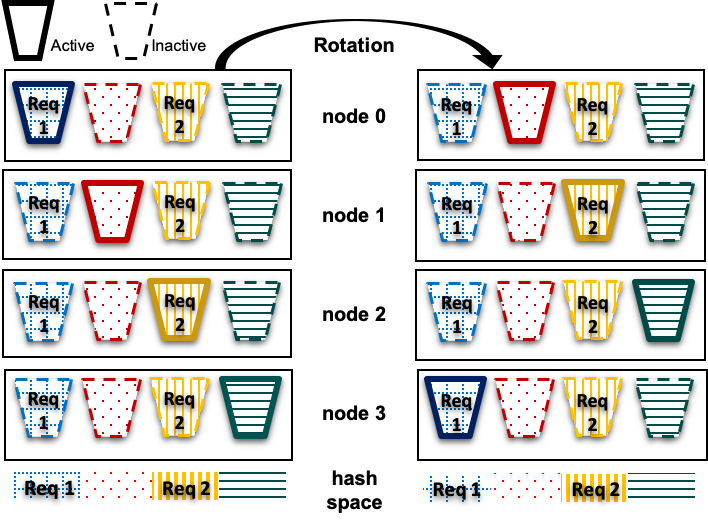}
		\caption{Request mapping in a \stable epoch with $n=4$ (all nodes are leaders):
			Solid lines represent the active buckets.
			Req. 1 is mapped to the first bucket, first active in node $1$.
			Req. 2 is mapped to the third bucket, first active in node $3$.
			Rotation redistributes bucket assignment across leaders.}
		\label{fig:buckets}
	\end{center}
\end{figure}

\paragraph{Buckets and Request Partitioning.}
To cope with these attacks,
    \protocol partitions the request hash space into buckets of equal size (number of buckets is a system parameter) and assigns each bucket to exactly one leader,
allowing a leader to only propose requests from its assigned (\emph{active}) buckets (preventing request duplication).
For load balancing, \protocol distributes buckets evenly
    (within the limits of integer arithmetics) to all leaders in each epoch.
To prevent request censoring, \protocol makes sure that
every bucket will be assigned to a correct leader infinitely often.
We achieve this by periodically redistributing the bucket assignment.
Bucket re-distribution happens
(1) at each epoch change (see Sec. \ref{sec:leaders-sns-buckets}) and
(2) after a predefined number of batches have been ordered in a stable epoch
    (since a stable epoch might never end), as illustrated in Figure~\ref{fig:buckets}.
Note that all nodes, while proposing only requests from their active buckets,
    still receive and store all requests (this can be optimized, see \ref{sec:client}).

\paragraph{Parallelism.}
The \protocol implementation (detailed in Sec.~\ref{sec:implementation}) is highly parallelized,
with every \emph{worker} thread responsible for one batch. In addition, \protocol uses multiple gRPC connections among each pair of nodes which proves to be critical in boosting throughput in a WAN especially with a small number of nodes.

\paragraph{Generalization of PBFT and Emulation of Other BFT Protocols.}
\protocol reduces to PBFT by setting $StableLeaders=1$. This makes every epoch stable, hides bucket rotation (primary is the single leader) and makes every epoch change ungracious. \protocol can also approximate protocols such as Tendermint \cite{Buchman16} and Spinning \cite{Spinning} by setting $StableLeaders>1$, and fixing the maximum number of batches and leaders in every epoch to 1, making every epoch an \ephemeral epoch and rotating leader/primary with every batch.


\section{\protocol Implementation Details}
\label{sec:details}

\subsection{The Client}\label{sec:client}
Upon $\bcast(r)$, broadcasting request $r$,
client $c$ creates a message $\langle\requestmsg, r,t,c\rangle_{\sigma_c}$.
The message includes the client's timestamp $t$, a monotonically increasing sequence number
that must be in a sliding window between the low and high \emph{client watermark} $t_{c_L}<t\le t_{c_H}$.
Client watermarks in \protocol allow multiple requests originating from the same client to be ``in-flight'',
to enable high throughput without excessive number of clients. These watermarks
are periodically advanced with the checkpoint mechanism described in Section~\ref{sec:checkpoint}, in a way which leaves no unused timestamps.
\add{Mir alligns checkpointing/advancing of the watermarks with epoch transitions (when no requests are in flight),
  such that all nodes always have a consistent view of the watermarks.}

In principle, the client sends the \requestmsg to all nodes (and periodically re-sends it to those nodes who have not received it, until request  commits at at least $f+1$ nodes).
In practice, a client may start by sending its request to fewer than $n$ nodes ($f+1$ in our implementation)
and only send it to the remaining nodes if the request has not been committed after a timeout.

\subsection{Sequence Numbers and Buckets}
\label{sec:leaders-sns-buckets}

\paragraph{Sequence Numbers.}
In each epoch $e$, a leader may only use a subset of $e$'s sequence numbers for proposing batches.
    \protocol partitions $e$'s sequence numbers to leaders in $\eplead(e)$ in a round-robin way, using modulo arithmetic,
starting at the epoch primary (see Fig.~\ref{fig:preprepare}).
We say that a leader \emph{leads} sequence number $sn$
when the leader is assigned $sn$ and is thus expected to send a \prepreparemsg for the batch with sequence number $sn$.
Batches are proposed in parallel by all epoch leaders and are processed like in PBFT.
Recall (from Table~\ref{table:diff}) that batch watermarking (not to be confused with client request watermarking from Sec. \ref{sec:client})
allows the PBFT primary to propose multiple batches in parallel;
in \protocol, we simply extend this to multiple leaders.

\begin{figure}[htbp]
    \begin{center}
        \includegraphics[width=\columnwidth]{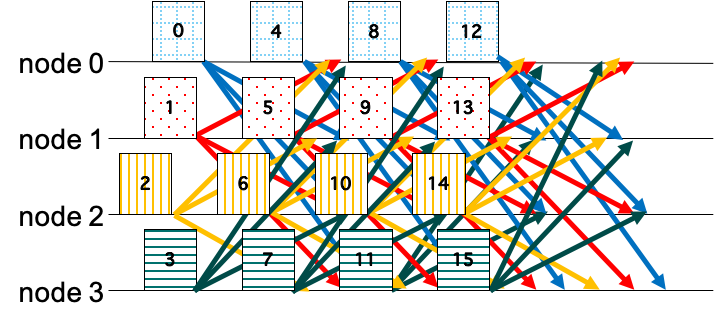}
        \caption{\prepreparemsg messages in an epoch
        where all 4 nodes are leaders balancing the proposal load. \protocol partitions batch sequence numbers among epoch leaders.}
        \label{fig:preprepare}
    \end{center}
\end{figure}

\paragraph{Buckets.}
In epoch $e=0$, we assign buckets to leaders sequentially,
starting from the buckets with the lowest hash values which we assign to epoch primary $0$.
For $e>0$, the primary picks a set of consecutive buckets for itself (primary's \emph{preferred buckets}),
starting from the bucket which contains the \emph{oldest} request it received;
this is key to ensuring Liveness (P5, Sec.~\ref{sec:model}).
\protocol distributes the remaining buckets evenly and deterministically among the other leaders --- this distribution is determined from an \emph{epoch configuration} which the epoch primary reliably broadcasts and which contains preferred buckets and leader set selection (see Sec.~\ref{sec:new-epoch}). Buckets assigned to a leader are called its \emph{active} buckets.

Additionally, if $e$ is stable (when $maxLen(e) = \infty$ and thus no further epoch changes are guaranteed),
leaders periodically (each time a pre-configured number of batches are committed)
rotate the bucket assignment:
leader $i$ is assigned buckets previously assigned to leader $i+1$ (in modulo $n$ arithmetic).
To prevent accidental request duplication
(which could result in leader $i$ being suspected and removed from the leader set),
leader $i$ waits to commit all ``in-flight'' batches
before starting to propose its own batches.
Other nodes do the same before pre-preparing batches in $i$'s new buckets.
In the example shown in Fig. \ref{fig:buckets}, after the rotation,
node 0 waits to commit all batches (still proposed by node 1) from its newly active red (second) bucket,
before node $0$ starts proposing new batches from the red (second) bucket.

\subsection{Common Case Operation}
\label{sec:stable}

\noindent\textbf{\requestmsg.} In the common case, the protocol proceeds as follows.
Upon receiving $\langle\requestmsg, r,t,c\rangle_{\sigma_c}$ from a client,
an epoch leader first verifies that the request timestamp is within the client's current watermarks $t_{C_L}<t\le t_{C_H}$
and maps the request to the respective bucket by hashing the client timestamp and identifier $h_r = H(t||c)$. Each bucket is implemented as a FIFO queue. We do \emph{not} hash the request payload, as this would allow a malicious client to target a specific bucket by adapting the request payload, mounting load imbalance attacks.
If the request falls into the leader's active bucket, the leader also verifies the client's signature on \requestmsg.
A node $i$ discards $r$ if $i$ already preprepared a batch containing $r$
or if $r$ is already pending at $i$
(we call a valid request pending if it has been received by $i$ but not yet committed).

\noindent\textbf{\prepreparemsg.}
Once  leader $i$ gathers enough requests in its current active buckets,
or if timer $T_{batch}$ expires (since the last batch was proposed by $i$),
$i$ adds the requests from the current active buckets in a batch,
assigns its next available sequence number $sn$ to the batch (provided $sn$ is within batch watermarks)
and sends a \prepreparemsg message.
If $T_{batch}$ time has elapsed and no requests are available,
$i$ sends a \prepreparemsg message with an empty batch.. This guarantees progress of the protocol under low load.

A node $j$ accepts a \prepreparemsg (we say \emph{preprepares} the batch and the requests it contains),
with sequence number $sn$ for epoch $e$ from node $i$ provided that:
(1) the epoch number matches the local epoch number and $j$ did not preprepare another batch with the same $e$ and $sn$,
(2) node $i$ is in $\eplead(e)$,
(3) node $i$ leads $sn$,
(4) the batch sequence number $sn$ in the \prepreparemsg is between a low watermark and high batch watermark: $w<sn\leq W$,
(5) none of the requests in the batch have already been \emph{preprepared},
(6-8) every request in the batch: (6) has timestamp within the current client's watermarks,
(7) belongs to one of $i$'s active buckets,
and (8) has a signature which verifies against client's id (i.e., corresponding public key).

Conditions (1)-(4) are equivalent to checks done in PBFT,
whereas conditions (5)-(8) differ from PBFT.
Condition (5) is critical for enforcing No Duplication (Property P3, Sec.~\ref{sec:model}).
Conditions (6) (allowing clients to send more than one request concurrently)
and (7) (prohibiting malicious leaders to propose requests outside their buckets) are performance related.
Condition (8) is the key to Validity (Property P1).
As this step may become a CPU bottleneck if performed by all nodes,
we use signature sharding as an optimization (see Sec.~\ref{sec:sharding}).

\noindent\textbf{The Rest.} If
node $j$ preprepares the batch, $j$ sends a \preparemsg and the protocol proceeds exactly as PBFT.
Otherwise, $j$ ignores the batch (which may  eventually lead to $j$ entering epoch change).
Upon committing a batch, $j$ removes all requests present in the committed batch from $j$'s pending queues.

\subsection{Epoch Change}
Locally, at node $j$, epoch $e$ can end \emph{graciously}, by exhausting all $maxLen(e)$ sequence numbers, or \emph{ungraciously}, if an epoch change timer (corresponding to the PBFT view change timer) at $j$ expires.
In the former (gracious) case, a node simply starts epoch $e+1$ (see also Sec. \ref{sec:new-epoch})
when it: (1) locally commits all sequence numbers in $e$, and (2) reliably delivers epoch configuration for $e+1$.
In the latter (ungracious) case, a node first enters an epoch change subprotocol (Sec. \ref{sec:ecsubprotocol}) for epoch $e+1$.

It can happen that some correct nodes finish $e$ graciously and some others not.
Such (temporary) inconsistency may prevent batches from being committed in $e+1$ even if the primary of $e+1$ is correct.
However, such such inconsistent epoch transitions are eventually resolved in subsequent epochs,
analogously to PBFT, when some nodes complete the view change subprotocol and some do not (due to asynchrony).
As we show in in
\ifdefined\FULLPAPER
  Appendix \ref{sec:livenessproof},
\else
  the anonymous full version~\cite{mirFull}
\fi
the liveness of \protocol is not violated.

\subsubsection{Epoch Change Subprotocol}
\label{sec:ecsubprotocol}

    The epoch change subprotocol is triggered by epoch timeouts due to asynchrony or failures
    and generalizes PBFT's view change subprotocol.
    Similarly to PBFT, after commiting batch $sn$ in epoch $e$
    a node resets and starts an epoch-change timer $ecTimer$ expecting to commit batch $sn+1$.

    If an $ecTimer$ expires at node $i$, $i$ enters the epoch-change subprotocol to move from epoch $e$ to epoch $e+1$.

    In this case, $i$ sends an \textsc{EPOCH-CHANGE} message to the primary of epoch $e+1$.
    An \textsc{EPOCH-CHANGE} message follows the structure of a PBFT \textsc{VIEW-CHANGE} message (page 411, \cite{Castro:2002:PBF})
    with the difference that it is signed and that there are no \textsc{VIEW-CHANGE-ACK} messages exchanged
    (to streamline and simplify the implementation similarly to \cite{Castro:1999:ABF}).
    The construction of a \textsc{NEW-EPOCH} message (by the primary of $e+1$)
    proceeds in the same way as the PBFT construction of a \textsc{NEW-VIEW} message.
    A node starts epoch $e+1$ by processing the \textsc{NEW-EPOCH} message
    the same way a node starts a new view in PBFT by processing a \textsc{NEW-VIEW} message.

    However, before entering epoch $e+1$,
    each correct node \emph{resurrects} potentially pre-prepared but uncommitted requests from previous epochs
    that are not reflected in the \textsc{NEW-EPOCH} message.
    This is required to prevent losing requests due to an epoch change  (due to condition (5) in pre-preparing a batch --- Sec.~\ref{sec:stable}),
    as not all batches that were created and potentially preprepared before the epoch change
    were necessarily delivered when starting the new epoch.
    Resurrecting a request involves each correct node:
    (1) returning these requests to their corresponding bucket queues, and
    (2) marking the requests as not preprepared.
    This allows proposing such requests again.
    Together with the fact that clients make sure that every correct replica eventually receives their request (Sec.~\ref{sec:client}),
    this helps guarantee Liveness (P5).

\subsubsection{Starting a New Epoch}
\label{sec:new-epoch}

Every epoch $e$ be it gracious or ungracious, starts by the primary reliably broadcasting
(using Bracha's classic 3-phase algorithm \cite{Bracha:1985:ACB:4221.214134})
the \emph{epoch configuration} information\footnote{We optimize the reliable broadcast of an epoch configuration using piggybacking on other protocol messages where applicable.}  containing:
(1) $\eplead(e)$, the set of epoch leaders for $e$, and
(2) identifiers of primary's preferred buckets (that the primary picks for itself),
which the primary selects based on the oldest request pending at the primary.

Before starting to execute participate in epoch $e$
(including processing a potential \textsc{NEW-EPOCH} message for $e$)
a node $i$ first waits to reliably deliver the epoch $e$ configuration.
In case of gracious epoch change, node $i$ also waits to locally committing all ``in-flight'' batches
pertaining to $e-1$.

\subsection{Checkpointing (Garbage Collection)} \label{sec:checkpoint}
Exactly as in PBFT, \protocol uses a checkpoint mechanism to prune the message logs.
After a node $i$ commits all batches with sequence numbers up to and including $sn_C$, where $sn_C$ is divisible by \add{predefined} configuration parameter $C$,
$i$ sends a $\langle\checkpointmsg, sn_C, H(sn_C')\rangle\sigma_i$ message to all nodes,
where $sn_C'$ is the last checkpoint and $H(sn_C')$ is the hash of the batches with sequence numbers $sn$ in range $sn_C'\le sn < sn_C$.
Each node collects checkpoint messages until it has $2\faults+1$ matching ones (including its own), constituting a \emph{checkpoint certificate},
and persists the certificate.
At this point, the checkpoint is \emph{stable}
and the node can discard the common-case messages from its log for sequence numbers lower than $sn$.

\protocol advances batch watermarks at checkpoints like PBFT does.
Clients' watermarks are also possibly advanced at checkpoints,
as the state related to previously delivered requests is discarded.
For each client $c$, the low watermark $t_{c_L}$ advances to the highest timestamp $t$ in a request submitted by $c$
that has been delivered, such that all requests with timestamp $t' < t$ have also been delivered.
The high watermark advances to $t_{c_H} = t_{c_L} + w_c$, where $w_c$ is the length of the sliding window.

\subsection{Signature Verification Sharding (SVS)}
\label{sec:sharding}

To offload CPU during failure-free execution (in \stable epochs),
we implement an optimization where not all nodes verify all client signatures.
For each batch, we distinguish $f+1$ \emph{verifier} nodes,
defined as the $f+1$ lexicographic (modulo $n$) successors of the leader proposing the batch.
Only the verifiers verify client signatures in the batch
on reception of a \prepreparemsg message (condition (8) in Sec. \ref{sec:stable}). Furthermore, we modify the \protocol (and thus PBFT) common-case protocol
such that a node does not send a \commitmsg before having received a \preparemsg message
from \emph{all} $f+1$ verifiers (in addition to $f$ other nodes and itself).
This maintains Validity, as at least one correct node must have verified client's signature. This way, however, if even a single verifier is faulty,
SVS may prevent a batch from being committed.
Therefore, we only apply this optimization in \stable epochs where all nodes are leaders.
In case (ungracious) epoch change occurs,
reducing the size of the leader set, \protocol disables SVS.

\subsection{State Transfer}\label{sec:statetransfer}

Nodes can temporarily become unavailable, either due to asynchrony, or due to transient failures.
Upon recovery/reconnection, nodes must obtain several pieces of information before being able to actively participate in the protocol again.
\protocol state transfer is similar to that of PBFT, and here we outline the key aspects of our implementation.

To transfer state, nodes need to obtain current epoch configuration information, the latest stable checkpoint (which occurred at sequence number $h$),
as well as information concerning batches having sequence numbers between $h+1$ and the latest sequence number.
Nodes also exchange information about committed batches..

The state must, in particular, contain two pieces of information:
(1) the current epoch configuration, which is necessary to determine the leaders from which the node should accept proposals, and
(2) client timestamps at the latest checkpoint, which are necessary to prevent including client requests that have already been proposed in future batches.

A node $i$ in epoch $e$ goes to state transfer when $i$ receives common-case messages from $f+1$ other nodes with epoch numbers higher than $e$, and $i$ does not transition to $e+1$ for a certain time.
Node $i$ obtains this information by broadcasting a $\langle HELLO, ne_i, c_i, b_i\rangle$ message,
where $ne_i$ is the latest NEW-EPOCH message received by $i$,
$c_i$ is the node's last stable checkpoint,
and $b_i$ is the last batch $i$ delivered.
Upon receipt of a $HELLO$ message, another node $j$ replies with its own $HELLO$ message,
as well as with any missing state from the last stable checkpoint and up to its current round $n$.

From the latest stable checkpoint, a node can derive the set of $2f+1$ nodes which signed this stable checkpoint. This also allows a node to transfer missing batches even from one out of these $2f+1$ nodes, while receiving confirmations of hashes of these batches from $f$ additional nodes (to prevent ingress of batches from a Byzantine node).

We perform further optimizations in order to reduce the amount of data that needs to be exchanged in case of a state transfer.
First, upon reconnecting, nodes announce their presence but wait for the next stable checkpoint after state transfer before actively participating in the protocol again. This enables us to avoid transferring the entire state related to requests following the preceding stable checkpoint.
Second, the amount of data related to client timestamps that needs to be transmitted can be reduced through only exchanging the root of the Merkle tree containing the client timestamps, with the precise timestamps being fetched on a per-need basis.

\subsection{Membership Reconfiguration} While details of membership reconfiguration are outside of the scope of this paper, we briefly describe how \protocol deals with adding/removing clients and nodes. Such requests, called \emph{configuration} requests are totally ordered like other requests, but are tagged to be interpretable/executed by nodes. As \protocol processes requests out of order (just like PBFT), configuration requests cannot be executed right after committing a request as the timing of commitment might diverge across nodes resulting in non-determinism. Instead, configuration requests are taken into account only at checkpoints and more specifically all configuration requests ordered between checkpoints $k-1$ and $k$, take effect only after checkpoint $k+1$.

\subsection{Durability (Persisting State)}

By default, \protocol implementation does not persist state or message logs to stable storage. Hence, a node that crashes might recover in a compromised state --- however such a node does not participate in the protocol until the next stable checkpoint which effectively restores the correct state. While we opted for this approach assuming that for few dozens of nodes simultaneous faults of up to a third of them will be rare, for small number of nodes the probability of such faults grows and with some probability might exceed threshold $\faults$. Therefore, we optionally persist state pertaining to \emph{sent} messages in \protocol, which is sufficient for a node to recover to a correct state after a crash.

We also evaluated the impact of durability with $4$ nodes, in a LAN setting, where it is mostly relevant due to small number of nodes and potentially collocated failures, using small transactions. We find that durability has no impact on total throughput, mainly due to the fact that persisted messages are amortized due to batching, \protocol parallel architecture and the computation-intensive workload. However, average  request latency increases by roughly 300ms.

  \subsection{Implementation Architecture}
  \label{sec:implementation}

  We implemented \protocol in Go.
  Our implementation is multi-threaded and inspired by  the \emph{consensus-oriented parallelism} (COP) architecture
  previously applied to PBFT to maximize its throughput on multicore machines \cite{BehlDK15}.
  Specifically, in our implementation, a separate thread is dedicated to managing each batch during the common case operation,
  which simplifies \protocol code structure and helps maximize performance.
  We further parallelize computation-intensive tasks whenever possible (e.g., signature verifications, hash computations).
  The only communication in common case between \protocol threads pertains to request duplication prevention
  (rule (6) in accepting \prepreparemsg in Sec.~\ref{sec:stable})
  --- the shared data structures for duplication prevention are hash tables,
  synchronized with per-bucket locks;
  instances that handle requests corresponding to different leaders do not access the same buckets.
  The only exception to the multi-threaded operation of \protocol is during an ungracious epoch-change,
  where a designated thread (\protocol Manager) is responsible for stopping worker common-case threads
  and taking the protocol from one epoch to the next.
  This manager thread is also responsible for sequential batch delivery and for checkpointing,
  which, however, does not block the common-case threads processing batches.

  Our implementation also parallelizes network access using a configurable number of independent network connections between each pair of nodes.  This proves to be critical in boosting \protocol performance beyond seeming bandwidth limitations in a WAN
  that stem from using a single TCP/TLS connection.

  In addition to multiple inter-node connections,
  we use an independent connection for handling client requests.
  As a result, the receipt of requests is independent of the rest of the protocol
  --- we can safely continue to receive client requests even if the protocol is undergoing an epoch change.
  Our implementation can hence seamlessly use, where possible,
  separate NICs for client's requests and inter-node communication to address DoS attacks \cite{Aardvark}.

\section{Pseudocode}
\label{sec:pseudocode}

In this section we introduce \protocol pseudocode. We first present PBFT~\cite{Castro:2002:PBF} pseudocode to demonstrate the common message flow in the common case of the two protocols.

Each node executes its own instance of the algorithm described by the pseudocode.
The node atomically executes each \textbf{upon} block exactly once for each assignment of values satisfying the block's triggering condition.

For better readability we do not include batching in the pseudocode.
Implementing batching is trivial by replacing requests with batches of requests,
except request handling (lines 53-60).
Moreover, whenever appropriate, instead of performing a request-specific action on a batch,
we perform this action on all requests in a batch, like request validity checks in \prepreparemsg (lines 74-77)
and request resurrection (lines 138-146).
In the context of request-specific validity checks, we consider the whole batch invalid if any of the contained requests fails its validity check.
Finally condition on line 62 should be replaced with checking if there exist enough requests for a batch.

Moreover, for readability, the pseudocode does not include a batch timeout which ensures that
even with low load leaders continuously send batches (even if empty)
to drive the checkpoint protocol and so that EpochChangeTimeout does not fire.

\begin{algorithm}[H]
  \scriptsize
  \begin{algorithmic}[1]

    \Function{IsPrimary}{i, v}
    \State \textbf{return} $i = v \mod N$;
    \EndFunction

    \Function{Valid}{v, n}
    \If {($lv=v$) \AndT ($w<=n<W)$}
    \State \textbf{return} $True$;
    \Else
    \State \textbf{return} $False$;
    \EndIf
    \EndFunction

    \Function{GetOldest}{S}
    \State Returns the oldest entry in set $S \setminus Preprepared$.
    \EndFunction

  \end{algorithmic}
  \caption{Common}
  \label{algorithm:common}
\end{algorithm}


\begin{algorithm}[H]
  \scriptsize
  \begin{algorithmic}[1]

    \Import Common
    \Import PbftViewChange
    \\
    \Parameters
    \State $id$ \Comment{The node identity}
    \State $f$ \Comment{Number of faults tolerated}
    \State $RequestTimeout$ \Comment{Timeout to prevent waiting indefinitely for q request to commit}
    \State $w$ \Comment{Low watermark, advances at checkpoints}
    \State $W$ \Comment{High watermark, advances at checkpoints}
    \EndParameters

    \Struct{Request}
    \State bytes $o$ \Comment{Request payload}
    \State int $t$ \Comment{Client timestamp}
    \State bytes $c$ \Comment{Client public key (ID)}
    \EndStruct

    \Init
    \State $lv\leftarrow 0$ \Comment{Local view number}
    \State $next\leftarrow 0$ \Comment{The next available sequence number}
    \State $R\leftarrow\emptyset$ \Comment{The set of received requests}
    \State $Preprepare\_msgs\leftarrow\{\}$ \Comment{A map from (view, sequence number) pairs to \prepreparemsg messages, initially $\bot$}
    \State $Prepare\_msgs\leftarrow\{\}$ \Comment{A map from (view, sequence number) pairs to a set of \emph{unique} \preparemsg messages}
    \State $Commit\_msgs\leftarrow\{\}$ \Comment{A map from (view, sequence number) pairs to a set of \emph{unique} \commitmsg messages}
    \State $RequestTimeouts\leftarrow\{\}$ \Comment{A map from requests to timers}
    \EndInit

  \algstore{pbft}
  \end{algorithmic}
  \caption{PBFT~\cite{Castro:2002:PBF}}
  \label{algorithm:pbft}
\end{algorithm}


  \begin{algorithm}[H]
  \scriptsize
  \ContinuedFloat
  \begin{algorithmic}[1]
  \algrestore{pbft}

    \State \textbf{upon receiving } $r \leftarrow \langle\requestmsg,o,t,c\rangle_{\sigma_c}$
    \State \hskip 0.55cm \textbf{ such that} $SigVer(r, \sigma_c, c)$
    \State \hskip 0.55cm \AndT \Not $(r^\prime \In R$ s.t. $r^\prime.c=r.c$ \AndT $r^\prime.t\ne r.t)$
    \textbf{do}
    \State \hskip \algorithmicindent $R\leftarrow r\cup R$
    \State \hskip \algorithmicindent $RequestTimeouts[r]\leftarrow$ \textbf{schedule} $RequestTimeout$
    \State

    \State \textbf{upon receiving } $|R|>0$ \AndT $w<= next <W $
    \State \AndT common.IsPrimary(id, lv)
    \textbf{do}
    \State $r\leftarrow common.GetOldest(R)$
    \State Send $\langle \prepreparemsg, lv, next, r, id\rangle$ to all nodes
    \State $next\leftarrow next+1$

    \State \textbf{upon receiving } $pp \leftarrow \langle \prepreparemsg, v, n, r, i\rangle$
    \State \hskip 0.55cm \textbf{ such that} $common.Valid(v, n)$
    \State \hskip 0.55cm \AndT $common.IsPrimary(i, v)$
    \State \hskip 0.55cm \AndT $Preprepare\_msgs[v, n] = \bot$
    \State \hskip 0.55cm \AndT $r \In R$ \textbf{do}
    \State \hskip \algorithmicindent $Preprepare\_msgs[v,n] \leftarrow pp$
    \State \hskip \algorithmicindent send $\langle \preparemsg, v, n, D(r), id\rangle$ to all nodes
    \State

    \State \textbf{upon receiving } $p \leftarrow \langle\preparemsg, v, n, D(r), i \rangle$
    \State \hskip 0.55cm \textbf{ such that} $D(Preprepare\_msgs[v, n].r) = D(r)$
    \State \hskip 0.55cm \AndT $common.Valid(n, v)$ \textbf{do}
    \State \hskip \algorithmicindent $Prepare\_msgs[v,n]\leftarrow Prepare\_msgs[v,n] \cup \{p\}$
    \State

    \Upon{|Prepare\_msgs[lv,n]| = 2f+1}
    \State $r \leftarrow Preprepare\_msgs[lv,n].r$
    \State send $\langle\commitmsg, lv, n, D(r), id\rangle$ to all nodes
    \EndUpon

    \State \textbf{upon receiving } $c \leftarrow \commitmsg, v, n, D(r), i\rangle$
    \State \hskip 0.55cm \textbf{ such that} $D(Preprepare\_msgs[v, n].r) = D(r)$
    \State \hskip 0.55cm \AndT $common.Valid(n, v)$ \textbf{do}
    \State \hskip \algorithmicindent $Commit\_msgs[v,n] \leftarrow Commit\_msgs[v,n] \cup \{c\}$
    \State

    \Upon{|Commit\_msgs[lv,n]| = 2f+1}
    \State $r\leftarrow Preprepare\_msgs[v, n].r$
    \State $R\leftarrow R\setminus r$
    \State $Deliver(n,r)$
    \State  \textbf{cancel} $RequestTimeouts[r]$
    \EndUpon

    \Upon{RequestTimeout}
    \State $lv\leftarrow lv+1$
    \State{$PbftViewChange.ViewChange()$}
    \Comment{Trigger standard PBFT view change}
    \EndUpon

  \end{algorithmic}
  \caption{PBFT (continues)}
  \label{algorithm:pbft2}
\end{algorithm}


\begin{algorithm}[H]
  \scriptsize
  \begin{algorithmic}[1]

    \Import Common
    \Import PbftViewChange
    \Import ReliableBroadcast
    \\
    \Parameters
    \State $id$ \Comment{The node identity}
    \State $f$ \Comment{Number of faults tolerated}
    \State $EpochChangeTimeout$ \Comment{Timeout for epoch change}
    \State $w$ \Comment{Low watermark, advances at checkpoints}
    \State $W$ \Comment{High watermark, advances at checkpoints}
    \State $NumBuckets$ \Comment{Number of buckets}
    \State $BucketsPerLeader$ \Comment{The number of buckets per leade when all nodes are leaders}
    \State $RotationPeriod$ \Comment{Bucket rotation period}
    \State $EphemeralEpLen$ \Comment{Number of sequence numbers in an ephemeral epoch}
    \EndParameters

    \Struct{Request}
    \State $o$ \Comment{Request payload}
    \State $t$ \Comment{Client timestamp}
    \State $c$ \Comment{Client struct}
    \EndStruct

    \Struct{Client}
    \State $pk$ \Comment{Client public key (ID)}
    \State $H$ \Comment{Client high watermark, advances at checkpoint}
    \State $L$ \Comment{Client low watermark, advances at checkpoint}
    \EndStruct

    \Struct{EpochConfig}
    \State $First$ \Comment{First sequence number of the epoch}
    \State $Last$ \Comment{Last sequence number of the epoch}
    \State $Leaders$ \Comment{List of leaders of the epoch}
    \State $PrimaryBuckets$ \Comment{Buckets the primary chose for itself}
    \EndStruct

    \Init
    \State $le\leftarrow 0$ \Comment{Local epoch number}
    \State $next\leftarrow id$ \Comment{The next available sequence number}
    \State $Buckets\leftarrow \text{ Set of }NumBuckets \text{ empty buckets}$
    \State \Comment{Each bucket is a FIFO queue of received requests}
    \State $Preprepare\_msgs\leftarrow\{\}$ \Comment{A map from (epoch, sequence number) pairs to \prepreparemsg messages}
    \State $Prepare\_msgs\leftarrow\{\}$ \Comment{A map from (epoch, sequence number) pairs to a set of \emph{unique} \preparemsg messages}
    \State $Commit\_msgs\leftarrow\{\}$ \Comment{A map from (epoch, sequence number) pairs to a set of \emph{unique}\commitmsg messages}
    \State $Preprepared\leftarrow\emptyset$ \Comment{A set of preprepared requests to prevent duplicates}
    \State $committed\leftarrow\{\}$ \Comment{A map from (epoch, sequence number) pairs to committed requests, initially $\bot$}
    \State $delivered\leftarrow\{\}$ \Comment{A map from (epoch, sequence number) booleans}
    \State $EpochConfig\leftarrow[]$ \Comment{List of epoch configurations}
    \ForAll{bucket}{Buckets}
    \State $bucket\leftarrow\emptyset$
    \EndForAll

    \State $EpochConfig[0].First = 0$
    \State $EpochConfig[0].Last = \infty$
    \State $EpochConfig[0].Leaders = Nodes$
    \State $EpochConfig[0].PrimaryBuckets =$ arbitrary $\lceil NumBuckets / Nodes \rceil$ buckets

    \State $ActiveBucketAssignment(0, EpochConfig[0])$
    \EndInit

    \algstore{mir}
  \end{algorithmic}
  \caption{\protocol initialization}
  \label{algorithm:mir1}
\end{algorithm}


\begin{algorithm}[H]
  \scriptsize
  \ContinuedFloat
  \begin{algorithmic}[1]
    \algrestore{mir}

    \State \textbf{upon receiving } $r \leftarrow\langle\requestmsg, o, t, c\rangle_{\sigma_c}$
    \State \hskip 0.55cm \textbf{ such that} $SigVer(m, \sigma_c, c.pk)$
    \State \hskip 0.55cm \AndT $r.c.L <= r.t < r.c.H$
    \State \hskip 0.55cm \AndT $r \notin Preprepared$
    \textbf{do}
    \State \hskip \algorithmicindent $bucket \leftarrow GetBucket(H(t||c.pk))$
    \State \hskip \algorithmicindent \textbf{if} $\nexists r' \in bucket : r'.c = r.c \wedge r'.t = r.t$
    \State \hskip \algorithmicindent \hskip \algorithmicindent $bucket.append(r)$
    \State \hskip \algorithmicindent \textbf{end if}
    \State

    \State \textbf{upon } $|ActiveBuckets(i,le,next)| > 0$
    \State \hskip 0.55cm \AndT $w <= n < W$
    \State \hskip 0.55cm \AndT $ActiveRotation(le,n)$
    \State \hskip 0.55cm \AndT $n \leq EpochConfig[le].Last$
    \textbf{do}
    \State \hskip \algorithmicindent $r\leftarrow common.GetOldest(ActiveBuckets(i,le,next))$
    \State \hskip \algorithmicindent send $\langle\prepreparemsg, le, n, r, id\rangle$ to all nodes
    \State \hskip \algorithmicindent $next \leftarrow next + |EpochConfig[le].Leaders)|$
    \State

    \State \textbf{upon receiving } $pp\leftarrow\langle\prepreparemsg, e, n, r, i\rangle$
    \State \hskip 0.55cm \textbf{ such that} $common.Valid(e, n)$
    \State \hskip 0.55cm \AndT $IsLeader(i, e, n)$
    \State \hskip 0.55cm \AndT $Preprepare\_msgs[e, n] = \bot)$
    \State \hskip 0.55cm \AndT $r.c.L <= r.t < r.c.H$
    \State \hskip 0.55cm \AndT $H(r.o||r.t||r.c.pk) \Not \In Preprepared$
    \State \hskip 0.55cm \AndT $H(r.t||r.c.pk) \In ActiveBuckets(i,e,n)$
    \State \hskip 0.55cm \AndT $SigVer((r, r.\sigma_c, c.pk)$
    \textbf{do}
    \State \hskip \algorithmicindent $Preprepared \leftarrow Preprepared \cup \{r\}$
    \State \hskip \algorithmicindent $Preprepare\_msgs[e, n] \leftarrow pp$
    \State \hskip \algorithmicindent send $\langle \preparemsg, v, n, D(r), id \rangle$ to all nodes
    \State

    \State \textbf{upon receiving } $p \leftarrow \langle\preparemsg, e, n, D(r), i\rangle$
    \State \hskip 0.55cm \textbf{ such that} $D(Preprepare\_msgs[e,n].r)=D(r)$
    \State \hskip 0.55cm \AndT $common.Valid(e, v)$
    \textbf{do}
    \State \hskip \algorithmicindent $Prepare\_msgs[e,n]\leftarrow Prepare\_msgs[e,n] \cup \{p\}$
    \State

    \Upon{|Prepare\_msgs[le,n]| = 2f+1}
    \State $r\leftarrow Preprepare\_msgs[e,n].r$
    \State send $\langle\commitmsg, le, n, D(r), id\rangle$ to all nodes
    \EndUpon

    \State \textbf{upon receiving } $c \leftarrow\langle\commitmsg, e, n, D(r), i\rangle$
    \State \hskip 0.55cm \textbf{ such that} $D(Preprepare\_msgs[e,n].r)=D(r)$
    \State \hskip 0.55cm \AndT $common.Valid(e, v)$
    \textbf{do}
    \State \hskip \algorithmicindent $Commit\_msgs[e,n]\leftarrow Commit\_msgs[e, n] \cup \{c\}$
    \State

    \Upon{|Commit\_msgs[e,n]| = 2f+1}
    \State $r \leftarrow Preprepare\_msgs[e,n].r$
    \State $committed[e, n] \leftarrow r$
    \State $GetBucket(H(r.t||r.c.pk)).remove(r)$
    \EndUpon

    \Upon{committed[le, n] \neq \bot \AndT delivered[n-1]}
    \State $Deliver(n,r)$
    \State $delivered[n]\leftarrow True$
    \State \textbf{reset} EpochChangeTimeout
    \EndUpon

    \Upon{EpochChangeTimeout}
    \State $PBFTViewChange.ViewChange()$
    \Comment{Standard PBFT view change}
    \EndUpon

    \State \textbf{upon } delivered[EpochConfig[e].Last]
    \State \AndT common.IsPrimary(id, e+1)
    \textbf{do}
    \State $EpochConfig[e+1].Leaders \leftarrow EpochConfig[e].Leaders \cup \{id\}$
    \State $EpochConfig[e+1].PrimaryBuckets$
    \State \hskip \algorithmicindent $\leftarrow \lceil NumBuckets / Nodes \rceil$ buckets containing the oldest requests
    \State $EpochConfig[e+1].First \leftarrow EpochConfig[e].Last + 1$
    \If{$EpochConfig[e+1].Leaders = Nodes$}
    \State $EpochConfig[e+1].Last \leftarrow \infty$
    \Else
    \State $EpochConfig[e+1].Last$
    \State \hskip \algorithmicindent $\leftarrow EpochConfig[e+1].First + ephemeralEpLen$
    \EndIf
    \State $ReliableBroadcast.Broadcast(EpochConfig[e+1], e+1)$

    \algstore{mir2}
  \end{algorithmic}
  \caption{\protocol (continues)}
  \label{algorithm:mir2}
\end{algorithm}


\begin{algorithm}[H]
  \scriptsize
  \ContinuedFloat
  \begin{algorithmic}[1]
    \algrestore{mir2}

    \Upon{\text{sending PBFT NEW-EPOCH message for epoch }e+1}
    \State $EpochConfig[e+1].Leaders \leftarrow ShrinkingLeaderset(e+1,id)$
    \State $EpochConfig[e+1].PrimaryBuckets$
    \State \hskip \algorithmicindent $\leftarrow \lceil NumBuckets / Nodes \rceil$ buckets containing the oldest requests
    \State $EpochConfig[e+1].First \leftarrow EpochConfig[e].Last + 1$
    \State $EpochConfig[e+1].Last \leftarrow EpochConfig[e+1].First + ephemeralEpLen$
    \State\label{ln:rb-ungracious} $ReliableBroadcast.Broadcast(EpochConfig[e+1], e+1)$
    \EndUpon

    \Upon{ReliableBroadcast.Delivered(EpochConfig, e) \AndT le = e}\label{ln:rb-deliver}
    \State $EpochConfig[e] \leftarrow EpochConfig$
    \If{$\exists k : EpochConfig[e].Leaders[k] = id$}
    \State $next\leftarrow EpochConfig[e].First+k$
    \EndIf
    \State $ActiveBucketAssignment(e, EpochConfig)$
    \EndUpon

    \Upon{\text{sending or receiving PBFT NEW-EPOCH message}}
    \ForAll{r}{Preprepared}
    \If{r \text{not in NEW-VIEW}}
    \State $Preprepared \leftarrow Preprepared \setminus \{r\}$
    \EndIf
    \EndForAll
    \ForAll{r}{\text{PBFT NEW-VIEW}}
    \State $Preprepared \leftarrow Preprepared \cup \{r\}$
    \EndForAll
    \EndUpon

    \Function{IsLeader}{i, e, n}
    \Comment{Returns $True$ if $i$ is leader of $n$ in epoch $e$}
    \If{$i \In EpochConfig[e].Leaders$}
    \State \textbf{return} $(EpochConfig[e].First + n \equiv i \mod |EpochConfig[e].Leaders|$
    \Else
    \State \textbf{return} False
    \EndIf
    \EndFunction

    \Function{GetBucket}{hash}
    \State Returns the bucket containing requests $r$ such that $H(r.t||r.c.pk)=hash$.
    \EndFunction

    \Function{ActiveBucketAssignment}{e, EpochConfig}
    \State Evenly partition $Buckets \setminus EpochConfig.PrimaryBuckets$
    \State among $EpochConfig.Leaders \setminus \{i: IsPrimary(i, e)\}$
    \EndFunction

    \Function{ActiveBuckets}{i,e,n}
    \State Returns the union of buckets which are active for node $i$ in epoch $e$ and sequence number $n$
    \EndFunction

    \State \Comment{ActiveRotation returns true if all the sequence numbers from the previous rotation are delivered}
    \Function{ActiveRotation}{e,n}
    \State $period\leftarrow RotationPeriod$
    \State $rotation\leftarrow \ceil{n-(EpochConfig[e-1].Last)/period}$
    \State \textbf{return} $delivered[EpochConfig[e-1].Last + (rotation-1)*period]$
    \EndFunction

    \Function{ShrinkingLeaderset}{e, i}
    \State $e_{last} \leftarrow$ the last epoch for which $i$ has the configuration
    \State $Leaders \leftarrow EpochConfig[e_{last}].Leaders \cup \{i\}$
    \State $RemovedLeaders \leftarrow$ a random set of $min((e'-e),1)$ nodes from $EpochConfig[e_{last}].Leaders\setminus\{i\}$
    \State \textbf{return} $Leaders \setminus RemovedLeaders$
    \EndFunction

  \end{algorithmic}
  \caption{\protocol (continues)}
  \label{algorithm:mir4}
\end{algorithm}


\section{\protocol Correctness}
\label{sec:correctness}

In this section we outline \protocol correctness proof, proving TOB properties as defined in Section~\ref{sec:model}. We pay particular attention to Liveness (Appendix~\ref{sec:livenessproof}), as we believe it is the least obvious out of four \protocol TOB properties to a reader knowledgeable in PBFT. Where relevant, we also consider the impact of the signature verification sharding (SVS) optimization (Sec.~\ref{sec:sharding}).

We define a function for assigning request sequence numbers to individual requests, output by \textsc{commit} indication of TOB as follows.  For batch with sequence number $sn\ge 0$ committed by a correct node $i$, let $S_{sn}$ be the total number of requests in that batch (possibly $0$). Let $r$ be the $k^{th}$ request that a correct node commits in a batch with sequence number $sn\ge 0$. Then, $i$ outputs \textsc{commit}$(k+\sum_{j=0}^{sn-1}S_{j},r)$, i.e., node $i$ commits $r$ with sequence number  $k+\sum_{j=0}^{sn-1}S_{j}$.

\subsection{Validity (P1)}

\noindent (P1) \textbf{Validity:} If a correct node commits $r$ then some client broadcasted $r$.\\

\emph{Proof (no SVS).}  We first show Validity holds, without signature sharding.
If a correct node commits $r$ then at least $n-f$ correct nodes sent \commitmsg for a batch which contains $r$, which includes at least $n-2f\ge f+1$ correct nodes (Sec.~\ref{ssec:PBFT}).
Similarly, if a correct node sends \commitmsg  for a batch which contains $r$, then at least $n-2f\ge f+1$ correct nodes sent \preparemsg after pre-preparing a batch which contains $r$ (Sec.~\ref{sec:stable}).
This implies at least $f+1$ correct nodes executed Condition (8) in Sec.~\ref{sec:stable} and verified client's signature on $r$ as correct. Validity follows. \qed\\

\emph{Proof (with SVS).} With signature verification sharding (Sec.~\ref{sec:sharding}),
clients' signatures are verified by at least $\faults+1$ \emph{verifier} nodes belonging to the leader set, out of which at least one is correct.
As no correct node sends \commitmsg before receiving \preparemsg from all $\faults+1$ \emph{verifier} nodes (Sec.~\ref{sec:sharding}),
no request which was not correctly signed by a client can be committed  --- Validity follows. \qed

\subsection{Agreement (Total Order) (P2)}

\noindent (P2)  \textbf{Agreement:} If two correct nodes commit requests $r$ and $r'$ with sequence number $sn$, then $r=r'$.

\begin{proof} Assume by contradiction that there are two correct nodes $i$ and $j$ which commit,
respectively, $r$ and $r'$ with the same sequence number $sn$, such that $r\neq r'$.
Without loss of of generality, assume $i$ commits $r$ with $sn$ before $j$ commits $r'$ with $sn$
(according to a global clock not accessible to nodes),
and let $i$ (resp., $j$) be the first correct node that commits $r$ (resp., $r'$) with $sn$.

By the way we compute request sequence numbers,
the fact that $i$ and $j$ commit different requests at the same (request) sequence number
implies they commit different batches with same (batch) sequence number.
Denote these different batches by $B$ and $B'$, respectively, and the batch sequence number by $bsn$.

We distinguish several cases depending on the mechanism by which $i$ (resp.,  $j$) commits $B$ (resp $B'$).
Namely, in \protocol, $i$ can commit $req$ contained in batch $B$ in one of the following ways (commit possibilities (CP)):
\begin{enumerate}[label=CP\arabic*]
\item by receiving a quorum ($n-\faults$) of matching \commitmsg messages in the common case of an epoch for a fresh batch $B$
(a fresh batch here is a batch for which a leader sends a \prepreparemsg message --- see Sec..~\ref{ssec:PBFT} and Sec.~\ref{sec:stable})
\item by receiving a quorum ($n-\faults$) of matching \commitmsg messages following an ungracious epoch change,
where \textsc{NEW-EPOCH} message contains $B$ (Sec.~\ref{sec:ecsubprotocol}),
\item via the state transfer subprotocol (Sec.~\ref{sec:statetransfer}).
\end{enumerate}

As $i$ is the first correct node to commit request $r$ with $sn$ (and therefore batch $B$ with $bsn$),
it is straightforward to see that $i$ cannot commit $B$ via state transfer (CP3).
Hence, $i$ commits $B$ by CP1 or CP2.

We now distinguish several cases depending on CP by which $j$ commits $B'$.
In case $j$ commits $B'$ by CP1 or CP2, since \protocol common case follows the PBFT common case,
and \protocol ungracious epoch change follows PBFT view change
--- a violation of Agreement in \protocol implies a violation of Total Order in PBFT.
A contradiction.

The last possibility is that $j$ commits $B'$ by CP3 (state transfer).
Since $j$ is the first correct node to commit $B'$ with $bsn$,
$j$ commits $B'$ after a state transfer from a Byzantine node.
However, since (1) \protocol \textsc{checkpoint} messages (see Sec.~\ref{sec:checkpoint})
which are the basis for stable checkpoints and state transfer (Sec.~\ref{sec:statetransfer}) are signed,
and (2) stable checkpoints contain signatures of $2f+1$ nodes including at least $f+1$ correct nodes,
$j$ is not the first correct node to commit $B'$ with $bsn$.
A contradiction.
\end{proof}

\subsection{No Duplication (P3)}

\noindent (P3) \textbf{No duplication:} If a correct node commits request $r$ with sequence numbers $sn$ and  $sn'$, then  $sn = sn'$.\\
\begin{proof}
No-duplication stems from the way \protocol prevents duplicate pre-prepares (condition (5) in accepting \prepreparemsg, as detailed in Sec.~\ref{sec:stable}).

Assume by contradiction that two identical requests $req$ and $req'$ exist
such that $req=req'$ and correct node $j$ commits $req$ (resp., $req'$) with sequence number $sn$ (resp,. $sn'$) such that $sn\neq sn'$.

Then, we distinguish the following exhaustive cases:
\begin{itemize}
	\item \emph{(i)} $req$ and $req'$ are both delivered in the same batch, and
	\item \emph{(ii)} $req$ and $req'$ are delivered in different batches.
\end{itemize}

In case \emph{(i)}, assume without loss of generality that $req$ precedes $req'$ in the same batch.
Then, by condition (5) for validating a \prepreparemsg (Sec.~\ref{sec:stable}),
no correct node preprepares $req'$ and all correct nodes discard the batch which hence cannot be committed, a contradiction.

In case \emph{(ii)} denote the batch which contains $req$ by $B$ and the batch which contains $req'$ by $B'$.
Denote the set of at least $n-\faults\ge 2\faults+1$ nodes that prepare batch $B$ by $S$ and the set of at least $n-\faults\ge 2\faults+1$ that prepare batch $B'$ by $S'$.
Sets $S$ and $S'$ intersect in at least $n-2\faults\ge \faults+1$ nodes out of which at least one is correct, say node $i$.
Assume without loss of generality that $i$ preprepares $B$ before $B'$.
Then, the following argument holds irrespectivelly of whether $i$ commits batch $B$ before $B'$, or vice versa:
as access to datastructure responsible for implementing condition (5) is synchronized with per-bucket locks (Sec.~\ref{sec:implementation})
and since $req$ and $req'$ both belong to the same bucket as their hashes are identical,
$i$ cannot preprepare $req'$ and hence cannot prepare batch $B'$ which cannot be delivered --- a contradiction. \qed \\

It is easy to see that signature verification sharding optimization does not impact the No-Duplication property.
\end{proof}

\subsection{Totality (P4)}
\label{sec:totalityproof}

\begin{lemma}\label{lemma:totality-seq-nos}
If a correct node commits a sequence number $sn$, then every correct node eventually commits $sn$.
\end{lemma}
\begin{proof}
  Assume, by contradiction, that a correct node $j$ never commits any request with $sn$.
  We distinguish 2 cases:
  \begin{enumerate}
    \item $sn$ becomes a part of a stable checkpoint of a correct node $k$. In this case, after GST,
    $j$ eventually enters the state transfer protocol similar to that of PBFT,
    transfers the missing batches from $k$, while getting batch hash confirmations from $f$ additional nodes that signed the stable checkpoint $sn$ belongs to (state transfer is outlined in Sec.~\ref{sec:statetransfer}).
    A contradiction.

    \item $sn$ never becomes a part of a stable checkpoint.
    Then, the start of the watermark window will never advance past $sn$, and all correct nodes,
    at latest when exhausting the current watermark window, will start infinitely many ungracious epoch changes
    without any of them committing any requests.
    Correct nodes will always eventually exhaust the watermark window, since even in the absence of new client
    requests, correct leaders periodically propose empty requests (see Section \ref{sec:pseudocode}).
    Infinitely many ungracious epoch changes without committing any requests, however,
    is a contradiction to PBFT liveness.
  \end{enumerate}
\end{proof}

\noindent (P4) \textbf{Totality:} If a correct node commits request $r$, then every correct node eventually commits $r$.\\
\begin{proof}
  Let $i$ be a correct node that commits $r$ with sequence number $sn$.
   Then, by (P2) Agreement, no correct node can commit another $r'\neq r$ with sequence number $sn$. Therefore, all other correct nodes will either commit $r$ with $sn$ or never commit $sn$.
   The latter is a contradiction to Lemma~\ref{lemma:totality-seq-nos}, since $i$ committed some request with $sn$, all correct nodes commit some request with $sn$.
   Totality follows.
\end{proof}

\subsection{Liveness (P5)}
\label{sec:livenessproof}
We first prove a number of auxiliary Lemmas and then prove liveness.
Line numbers refer to \protocol pseudocode (Sec.~\ref{sec:pseudocode}).

\begin{lemma}\label{lemma:laststable}
	In an execution with a finite number of epochs, the last epoch $e_{last}$ is a stable epoch.
\end{lemma}
\begin{proof}
	Assume by contradiction that $e_{last}$ is not stable, this implies either:
	\begin{enumerate}
		\item a gracious epoch change from $e_{last}$ at some correct node and hence, $e_{last}$ is not the last -- a contradiction; or
		\item ungracious epoch change from $e_{last}$ never completes
		--- since \protocol ungracious epoch change protocol follows PBFT view change protocol,
		this implies liveness violation in PBFT.
		A contradiction.
	\end{enumerate}
\end{proof}

\begin{lemma}\label{lemma:correct-client}
	If a correct client broadcasts request $r$, then every correct node eventually receives $r$ and puts it in the respective bucket queue.
\end{lemma}
\begin{proof}
	By assumption of a synchronous system after GST,
    and by the correct client sending and periodically re-sending request to all nodes until a request is committed (see Section~\ref{sec:client}).
\end{proof}

\begin{lemma}\label{lemma:infinity}
	In an execution with a finite number of epochs,
	any request a correct node has received which no correct node commits before the last epoch $e_{last}$,
	is eventually committed by some correct node.
\end{lemma}
\begin{proof}

	Assume that a correct node $i$ has received some request $r$ that is never committed by some correct node $j$.
	Since (by  Lemma~\ref{lemma:laststable}) $e_{last}$ is an (infinite) stable epoch, 
	$i$ is the leader in $e_{last}$ and $i$ will propose infinitely often
	exhausting all available requests until $r$ is the oldest uncommitted request.
	Next time $i$ proposes a batch from $r$'s bucket (this eventually occurs due to bucket rotation in a stable epoch),
	if no other node proposed a batch with $r$, $i$ includes $r$ in its proposed batch with some sequence number $sn$.

	Assume some node $j$, also a leader in $e_{last}$ (since all nodes are leaders in a stable epoch),
	never commits any batch with sequence number $sn$.
	Then, epoch change timeout fires at $j$ and $j$ does not propose more batches.
	Eventually, epoch change timeout fires at all other correct nodes causing an epoch change.
	A contradiction (that $j$ never commits any batch with sequence number $sn$).

	Therefore, $j$ commits a batch with sequence number $sn$. By (P2) Agreement the batch $j$ commits is the same as the batch $i$ commits and contains $req$. A contradiction (that $j$ never commits $req$).
\end{proof}

\begin{lemma}[Liveness with Finitely Many Epochs]\label{lemma:unicorn}
	In an execution with a finite number of epochs,
	if a correct client broadcasts request $r$,
	then some correct node eventually commits $r$.
\end{lemma}
\begin{proof}
	By Lemma~\ref{lemma:laststable}, $e_{last}$ is a stable epoch.
	There are 2 exhaustive possibilities.
	\begin{enumerate}

    \item Some correct node commits $r$ in epoch preceding $e_{last}$.
    \item $r$ is not committed by any correct node before $e_{last}$.
		By Lemma~\ref{lemma:correct-client} all correct nodes eventually receive request $r$ and by Lemma~\ref{lemma:infinity} some correct node commits $r$.
		The lemma follows.
	\end{enumerate}
\end{proof}

\begin{definition} [Preferred request]
  \label{def:preferred-request}
	Request $r$ is called \emph{preferred request} in epoch $e$,
	if $r$ is the oldest request pending in buckets of primary of epoch $e$,
	before the primary proposes its first batch in $e$.
\end{definition}

\begin{lemma}\label{lemma:synchrony}
  If, after GST, all correct nodes start executing the common-case protocol in a non-stable epoch $e$ before time $t$,
  then there exists a $\Delta$, such that if a correct leader proposes a request $r$ before time $t$
  and no correct node enters a view change before $t + \Delta$,
  every correct node commits $r$.
\end{lemma}
\begin{proof}
  Let $\delta$ be the upper bound on the message delay after GST and let a correct leader propose a request $r$ before $t$.
  By the common-case algorithm without SVS (there is no SVS in a non-stable epoch $e$) all correct nodes receive at least $2f + 1$ \commitmsg messages for $r$
  before $t + 3\delta$ (time needed to transmit \prepreparemsg, \preparemsg and \commitmsg).
  All correct nodes will accept these messages, since they all enter epoch $e$ by time $t$.
  As every correct node receives at least $2f + 1$ \commitmsg{}s, every correct node commits $r$ by $t + 3\delta + rd$.
  Therefore, $\Delta = 3\delta + rd$.
\end{proof}

\begin{lemma}\label{lemma:ungracious-config}
    If all correct nodes perform an ungracious epoch change from $e$ to $e+1$ and the primary of $e+1$ is correct,
    then all correct nodes reliably deliver the epoch configuration of $e+1$.
\end{lemma}
\begin{proof}
    Let $p$ be the correct primary of $e+1$.
    As $p$ is correct, by the premise, $p$ participated in the ungracious epoch change subprotocol.
    Since the PBFT view change protocol is part of the ungracious view change, $p$ sends a \textsc{NEW-EPOCH} message to all nodes.
    By the algorithm (line \ref{ln:rb-ungracious}), $p$ reliably broadcasts the epoch configuration of $e+1$.
    Since all correct nodes participate in the ungracious view change, all correct nodes enter epoch $e+1$.
    By the properties of reliable broadcast,
    all correct nodes deliver the epoch configuration in $e+1$ (line \ref{ln:rb-deliver}).
\end{proof}

\begin{lemma}\label{lemma:ungracious-commit}
  There exists a time after GST, such that if each correct node reliably delivers the configuration of epoch $e$ after entering $e$ through an ungracious epoch change,
  and the primary of $e$ is correct, then all correct nodes commit $e$'s preferred request $r$.
\end{lemma}
\begin{proof}
  Let $p$ be the primary of epoch $e$, and $C$ the epoch configuration $p$ broadcasts for $e$.
  By the algorithm, the leader set $C$ does not contain all nodes (and thus SVS is disabled in $e$),
  as all correct nodes entered $e$ ungraciously.
  Since (by the premise) all correct nodes deliver $C$,
  all correct nodes will start participating in the common-case agreement protocol in epoch $e$.
  Let $t_f$ and $t_l$ be the time when, respectively, the first and last correct node does so.

  By the algorithm, $p$ proposes $r$ immediately when entering epoch $e$, and thus at latest at $t_l$.
  Then, by Lemma \ref{lemma:synchrony}, there exists a $\Delta$ such that all correct nodes commit $r$
  if none of them initiates a view change before $t_l + \Delta$.

  Eventually, after finitely many ungracious epoch changes, where all correct nodes double their epoch change timeout values (as done in PBFT \cite{Castro:2002:PBF}),
  all correct nodes' epoch change timeout will be greater than $(t_l - t_f) + \Delta$.
  Then, even if a node $i$ enters epoch $e$ and immediately starts its timer at $t_f$, $i$ will not enter view change
  before $t_l + \Delta$ and thus all correct nodes will deliver $r$ in epoch $e$.
\end{proof}

\begin{lemma}\label{lemma:ungracious}
   There exists a time after GST, such that if all correct nodes perform an ungracious epoch change from $e$ to $e+1$,
    and the primary of $e+1$ is correct,
    then some correct nodes commits preferred request in $e+1$.
\end{lemma}
\begin{proof}
Follows from Lemmas \ref{lemma:ungracious-config} and \ref{lemma:ungracious-commit}.
\end{proof}

\begin{lemma}\label{lemma:two-correct}
In an execution with infinitely many epochs there exist an infinite number of pairs of consecutive epochs with correct primaries.
\end{lemma}
\begin{proof}
  Epoch primaries succeed each other in a round robin way across all the lexicographically ordered nodes of the system (see Sec.~\ref{sec:overview} and Sec.~\ref{sec:model}).
  Assume such pair of two consecutive epochs with correct primaries never exists after some epoch $e$.
  Then, in every full rotation across all $3f+1$ nodes after $e$,
  there exist an epoch with a faulty primary node between every two epochs with correct primaries,
  which implies the number of faulty nodes to be greater than $\faults$.
  A contradiction.
\end{proof}

\begin{lemma}\label{lemma:progress}
There exists a time after GST, such that for any pair of consecutive epochs $e$ and $e+1$ with correct primaries $i$ and $j$ (respectively),
some correct node commits at least one of the preferred requests in $e$ and $e+1$.
\end{lemma}
\begin{proof}
  Let $r_e$ (resp., $r_e'$) be preferred request in $e$ (resp., $e'$).
  For the epoch change from $e$ to $e+1$ there are two exhaustive possibilities.
  \begin{enumerate}
    \item \emph{At least one correct node performs a gracious epoch change from $e$ to $e+1$.}
    Recall that \protocol requires the primary of an epoch to be in the leader set (Section~\ref{sec:overview}).
    As $e$ graciously ends at at least one correct node,
    it follows from the specification of the gracious epoch change (Section~\ref{sec:details}), that at least one node commits all requests proposed in $e$.

    Since, by the protocol, the primary of $e$ is in the leader set of $e$ and the correct primary always proposes the preferred request,
    at least one correct node commits the preferred request of $e$.
    \item \emph{No node performs a gracious epoch change.} By Lemma~\ref{lemma:ungracious}.
  \end{enumerate}
\end{proof}

\noindent (P5) \textbf{Liveness:} If a correct client broadcasts request $r$, then some correct node eventually commits $r$.\\

\begin{proof}
	We distinguish two cases:
	\begin{enumerate}
		\item In an execution with a finite number of epochs, Liveness follows from Lemma~\ref{lemma:unicorn}.

        \item Consider now an execution with an infinite number of epochs. By Lemma \ref{lemma:correct-client}, every correct node eventually receives $r$.
        Let $P$ be the set of all requests that some correct node received before it received $r$.
        After $r$ has been received by all correct nodes, following from Definition \ref{def:preferred-request},
        if $r' \neq r$ is a preferred request, then $r' \in P$.
        By Lemma \ref{lemma:progress}, however, all such requests $r'$ will eventually be committed by all correct nodes.
        Therefore, by Definition \ref{def:preferred-request}, unless $r$ is committed earlier by some correct node,
        $r$ will eventually become the preferred request of all epochs with correct primaries, and will be committed by some correct node by Lemma \ref{lemma:progress}.
	\end{enumerate}
\end{proof}


\section{LTO: Optimization for large requests}
\label{sec:LTO}

When the system is network-bound (e.g., with large requests and/or on a WAN)
the maximum throughput is driven by the amount of data each leader can send in a \prepreparemsg message.
However, data, i.e., request payload, is not critical for total order,
as the nodes can establish total order on request hashes.
While in many blockchain systems all nodes need data \cite{BitcoinURL,EthereumURL},
in others \cite{AndroulakiBBCCC18}, ordering is separated from request execution
and full payload replication across ordering nodes is unnecessary.

For such systems,
\protocol optionally boosts throughput using what we call \emph{Light Total Order (LTO)} broadcast.
LTO is defined in the same way as TO broadcast (Sec.~\ref{sec:model})
except that LTO requires property $P4$ (Totality) to hold for the hash of the request $H(r)$ (instead for request $r$).

\begin{enumerate}[label=P4]
    \item \textbf{Totality:} If a correct node commits a request $r$ or a request hash $H(r)$, then every correct node eventually commits $H(r)$.\\
\end{enumerate}

LTO follows SVS high-level approach (Sec.~\ref{sec:sharding}) and applies to \protocol only in stable epoch. A leader only sends a full \prepreparemsg message to a subset of $f + 1$ \emph{replica} nodes.
To the remaining $2f$ \emph{observer} nodes, the leader sends a lightweight \prepreparemsg message
where request payloads are replaced with their hashes.
Since inside the \protocol (and PBFT) common-case (Sec.~\ref{sec:stable}) subprotocol, before sending a \commitmsg message,
a node waits to receive $2f + 1$ \preparemsg messages,
this implies that at least $f + 1$ of \prepreparemsg messages were sent by replica nodes,
ensuring that at least one correct (replica) node has the full payload.

  LTO has minor impact on PBFT view change (\protocol epoch change) as a new primary might have a hash of the batch (lightweight \prepreparemsg) but not the full batch payload. To this end, our \protocol-LTO makes the primary in this situation look for the payload at $f+1$ replicas, which is guaranteed not to block liveness after GST and with the correct primary.


\begin{table}[ht]
	\centering
	\scriptsize{
		\begin{tabular}{| l | l | }
			\hline
			Max batch size & $2$ MB (4000 requests) \\ \hline
			Cut batch timeout & $500$ ms ($n<49$), $1s (n=49)$,\\
      &  $2s (n=100)$ \\\hline
			Max batches \ephemeral epoch &  256 ($n\leq16$), $16*n$ ($n>16$)\\ \hline
			Bucket rotation period & 256 ($n\leq16$), $16*n$ ($n>16$)  \\ \hline
			Buckets per leader ($m$) & $2$ \\ \hline
			Checkpoint period & $128$ \\ \hline
			Watermark window size & $256$  \\ \hline
			Parallel gRPC connections& $5$ ($n=4$), $3$ ($n=10$), $1$  ($n>10$) \\ \hline
			Client signatures & 256-bit ECDSA \\ \hline
		\end{tabular}
	}
	\caption{\protocol configuration parameters used in evaluation}
	\label{table:config}
\end{table}

\section{Evaluation}\label{sec:evaluation}

In this section, we report on experiments we conducted in scope of \protocol performance evaluation, which
 aims at answering the following questions:\\
\noindent (Sec.~\ref{sec:WANeval}) How does \protocol scale on a WAN?\\
\noindent (Sec.~\ref{sec:LANeval}) How does \protocol scale in clusters?\\
\noindent (Sec.~\ref{sec:OPTeval}) What is the impact of optimizations (SVS, LTO) and bucket rotation and what are typical latencies of \protocol?\\
\noindent (Sec.~\ref{sec:duplicates}) What is the benefit of \protocol duplication prevention?\\
\noindent (Sec.~\ref{sec:evalfaults}) How does \protocol perform under faults and 
attacks (crash faults, censoring attacks, straggler attacks)?\\

\begin{figure*}
	\begin{subfigure}{\columnwidth}
		\centering
		\includegraphics[width=\columnwidth]{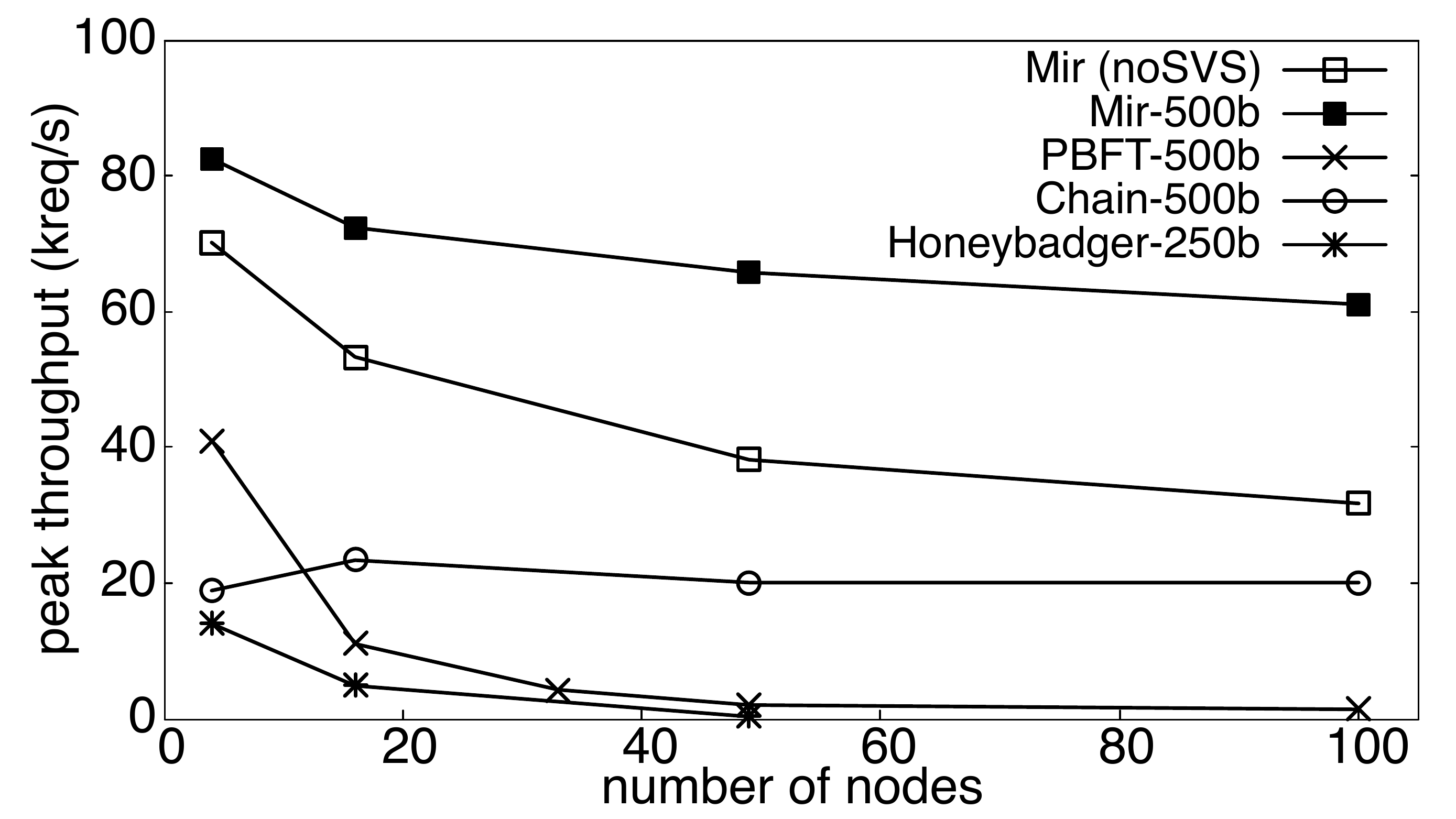}
		\caption{Mir vs. Chain, PBFT, Honeybadger.}
		\label{fig:500WAN}
	\end{subfigure}%
	\begin{subfigure}{\columnwidth}
		\centering
		\includegraphics[width=\linewidth]{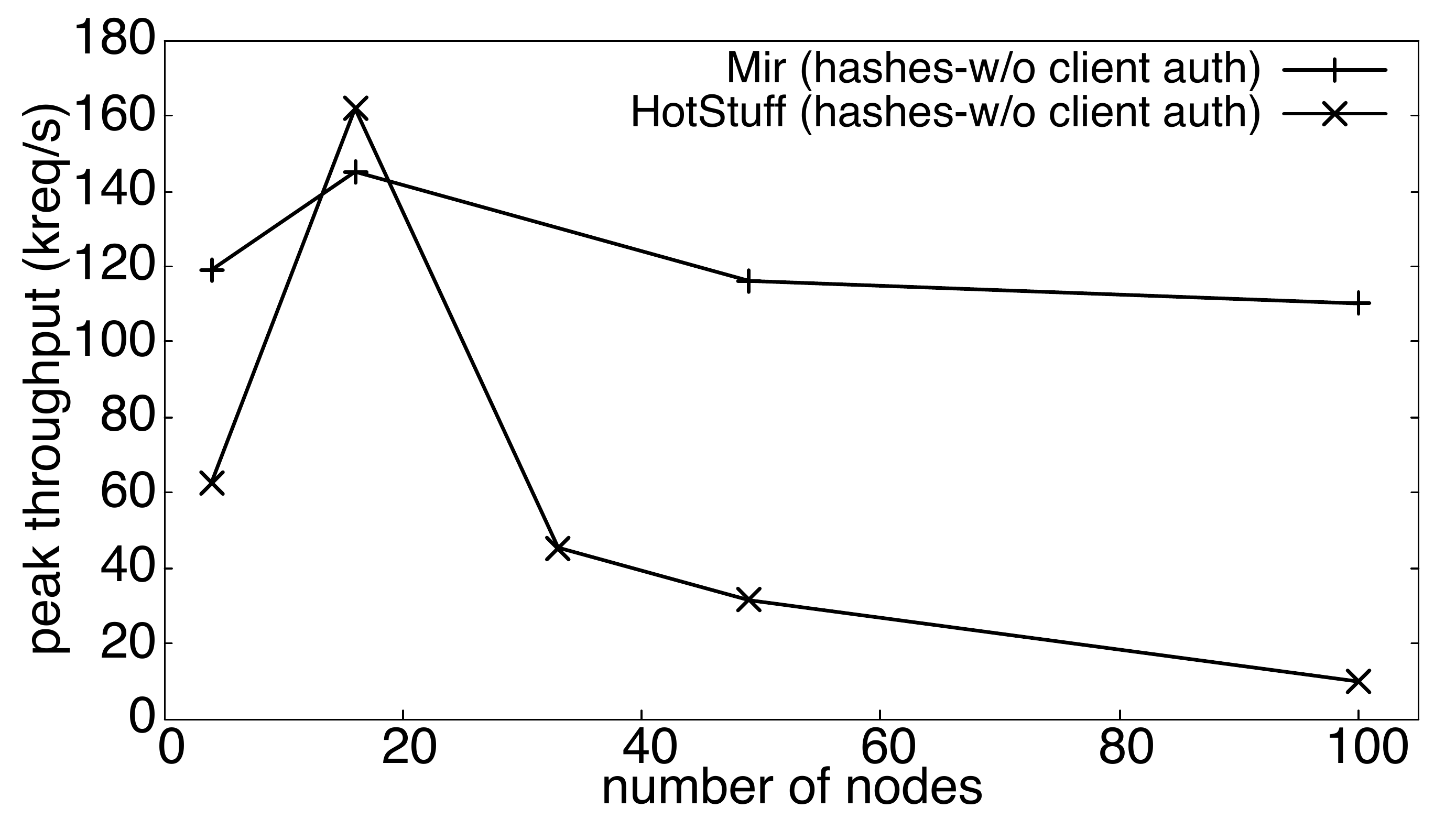}
		\caption{Mir vs. Hotstuff (no client authentication, only hashes ordered).}
		\label{fig:HSWAN}
	\end{subfigure}
	\caption{WAN scalability experiments.}
	\label{fig:WAN}
\end{figure*}

\noindent\textbf{Experimental Setup.}
Our evaluation consists of microbenchmarks of
500 byte requests, which correspond to average Bitcoin tx size \cite{BitcoinTXSizeURL}.
These are representative of \protocol performance,
both absolute and relative to state of the art.
We also evaluate \protocol in WAN for larger 3500 byte requests,
typical in Hyperledger Fabric \cite{AndroulakiBBCCC18} to better showcase the impact of available bandwidth on \protocol.

We generate client requests
by increasing the number of client processes
and the request rate per client process, until the throughput is saturated.
We report the throughput just below saturation.
The client processes estimate which node $i$ has an active bucket for each of their requests
and initially send each request only to nodes $i-1, \cdots, i+k$, where $k\leq f-1$, i.e.,. to $f+1$ nodes.

We compare \protocol to a state-of-the-art PBFT implementation optimized for multi-cores \cite{BehlDK15}.
For fair comparison, we use the \protocol codebase tuned to closely follow the PBFT implementation of \cite{BehlDK15}
hardened to implement Aardvark \cite{Aardvark}.
As another baseline, we compare the common case performance of Chain,
an optimistic subprotocol of the Aliph BFT protocol \cite{700} with linear common-case message complexity,
which is known to be near throughput-optimal in clusters,
although it is not robust and needs to be abandoned in case of faults \cite{700}.
In this sense, Chain is not a competitor to \protocol, but rather an upper bound on performance in a cluster.
PBFT and Chain are always given best possible setups,
i.e., PBFT leader is always placed in a node that has most effective bandwidth
and Chain spans the path with the smallest latency.
We further compare to HotStuff \cite{hotstuff} (a recent, popular, $O(n)$ common-case message complexity BFT protocol)
and Honeybadger \cite{MillerXCSS16} using their open source implementations \footnote{\url{https://github.com/hot-stuff/libhotstuff} at commit 978f39f... and \url{https://github.com/initc3/HoneyBadgerBFT-Python}}.
We present comparison to HotStuff separately, due to its implementation specifics.
We allow Honeybadger an advantage with using 250 byte requests, as its
open source implementation is fixed to this request size.
We do not compare to unavailable (e.g., Hashgraph \cite{Hashgraph}, Red Belly \cite{RedBelly}) or unmaintained (BFT-Mencius\footnote{We,  however, demonstrate the expected effective throughput of Hashgraph, Red-Belly and BFT-Mencius under request duplication, by ``switching off'' request duplication prevention in \protocol, see Sec.~\ref{sec:duplicates}.} \cite{BFT-Mencius}) protocols,
those faithfully approximated by PBFT (e.g., BFT-SMaRt \cite{BessaniSA14}, Spinning \cite{Spinning}, Tendermint \cite{TendermintURL}),
or those that report considerably worse performance than Mir (e.g., Algorand \cite{Algorand}).

We use virtual machines on IBM Cloud,
with 32 x 2.0 GHz VCPUs and 32GB RAM,
equipped with 1Gbps networking and limited to that value for experiment repeatability, due to non-uniform bandwidth overprovisioning we sometimes experienced.
~Table~\ref{table:config} shows the used \protocol configuration parameters.
Unless stated otherwise, \protocol uses signature sharding optimization.

\subsection{Scalability on a WAN}
\label{sec:WANeval}
To evaluate \protocol scalability, we ran it with up to $n=100$ nodes
on a WAN setup spanning 16 distinct datacenters across Europe, America, Australia, and Asia
Beyond $n=16$, we collocate nodes across already used datacenters. Our 4-node experiments spread over all 4 mentioned continents.
Client machines are also uniformly distributed across the 16 datacenters.
Figure~\ref{fig:WANmap} shows the datacenter distribution.
\begin{figure*}[tb!]
	\centering
	\includegraphics[width=0.85\textwidth,clip=true,trim=0 10 12 10]{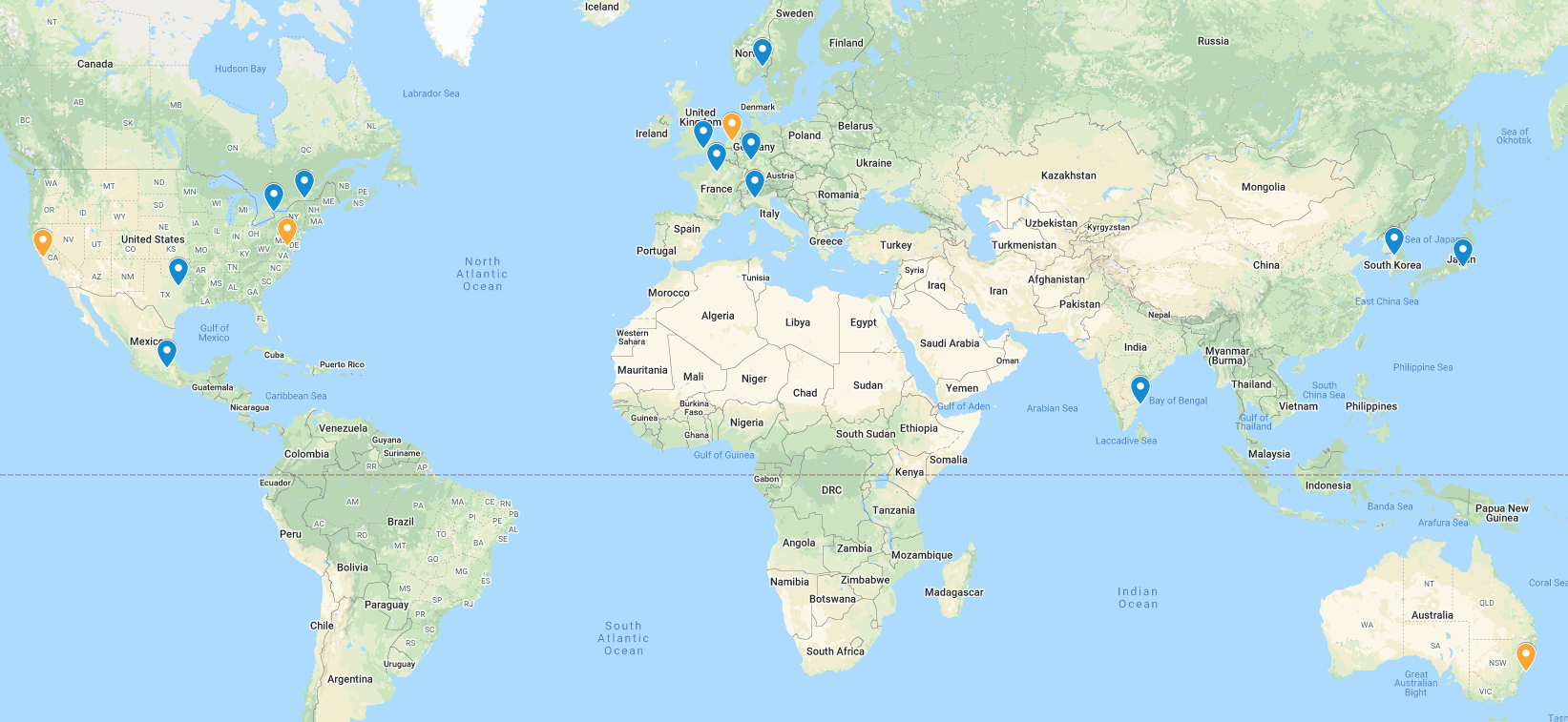}
	\caption{Distribution of the $16$ datacenters for WAN deployment. Yellow pins indicate the $n=4$ deployment.}
	\label{fig:WANmap}
\end{figure*}

Figure~\ref{fig:500WAN} depicts the common-case (failure-free) stable epoch performance of \protocol,
compared to that of PBFT and Chain
and Honeybadger.
We observe that PBFT throughput decays rapidly,
following an $O(1/n)$ function and scales very poorly.
Chain scales better and even improves with up to $n=16$ nodes,
sustaining 20k req/s,
but is limited by the bandwidth of the ``weakest link'', i.e., a TCP connection with lowest bandwidth across all links between consecutive nodes.
Compared to Honeybadger, \protocol retains much higher throughput, even though:
(i) Honeybadger request size is smaller (250 bytes vs 500 bytes),
and (ii) Honeybadger batches are significantly larger (up to 500K requests in our evaluation).
\ifdefined\SUBMIT
\else
This is due to the fact that Honeybadger is computationally bound by $O(n^2)$ threshold signatures verification and on top of that the verification of the signatures is done sequentially. Honeybadger's throughput also suffers from request duplication (on average $1/3$ duplicate requests per batch), since the nodes choose the requests they add in their batches at random.
Moreover, we report on Honeybadger latency, which is in the order of minutes (partly due to the large number of requests per batch and partly due to heavy computation), significantly higher than that of Mir.
In our evaluation we could not increase the batch size as much as in the evaluation in \cite{MillerXCSS16}, especially with increasing the number of nodes beyond $16$, due to memory exhaustion issues. Finally, in our evaluation PBFT outperforms Honeybadger (unlike in \cite{MillerXCSS16}), as our implementation of PBFT leverages the parallelism of \protocol codebase.
\fi

\protocol dominates other protocols, delivering 82.5k (roughly 4x the throughput of Chain)
with $n=4$.
With $n=100$, \protocol maintains more than 60k req/s (3x Chain throughput).
Even without the signature verification sharding otptimization (``Mir (noSVS)'') \protocol significantly outperforms other protocol, delivering with $n=4$ 70.2k req/s (3.5x Chain throughput) while reaching 31.7k req/s with $n=100$ (1.5x Chain throughput).

\noindent\textbf{Comparison to HotStuff in WANs.}
We present the comparison of \protocol to the HotStuff~\cite{hotstuff} leader-based protocol separately, in Figure~\ref{fig:HSWAN}.
Despite HotStuff specifying that the leader disseminates the request payload \cite{hotstuff},
the available HotStuff implementation orders only hashes of requests,
relying optimistically on clients for payload dissemination.
This approach is vulnerable to liveness/performance attacks from malicious clients
which can be easily mounted by clients not sending the requests to all nodes (an attack which the HotStuff version we evaluated does not address).
Besides, the evaluated HotStuff implementation did not authenticate clients at all (which jeopardizes Validity).

For these reasons and for a fair comparison,
we perform an experiment with: 1) disabled \protocol client authentication (i.e,. client signature verification)
and 2) with leaders disseminating payload hashes 
(relying on clients to disseminate payload as in HotStuff).
We also increase batch sizes in HotStuff as much as needed, resulting in up to 32K requests per batch, to saturate the system.

We observe that HotStuff offers about 2x lower throughput than \protocol with $n=4$ nodes
bounded by the number of available network connections,
whereas \protocol uses multiple connections among pairs of nodes.
As $n$ and number of network connections from the leader grow,
HotStuff throughput first grows until the network at the leader is saturated
(with $n=16$ HotStuff performs about 10\% better than \protocol).
However, as leader bandwidth becomes the bottleneck even with hash-only ordering,
HotStuff's $O(n^{-1})$ network-bound scalability starts to show with $n>16$,
while \protocol continues to scale well and is only computationally bounded by the implementation.
With 100 nodes, \protocol orders 110k hashes per second, compared to roughly 10k hashes per second throughput of HotStuff.

\noindent\textbf{Experiments with 3500-byte payload}
With request payload size (500 bytes), CPU related to signature verification is the primary bottleneck.
It is therefore interesting to  evaluate the impact on performance with larger requests.
Intuitively, with larger requests, we would be able to stress the 1Gbps WAN bottleneck of our evaluation testbed.
Moreover, large requests are not only of theoretical importance, some prominent blockchain systems feature relatively large transaction sizes.
For instance, minimum size transaction in Hyperledger Fabric is about 3.5kbytes \cite{AndroulakiBBCCC18}.

Therefore, we conducted additional WAN experiments with 3500 bytes request size.

\begin{figure}[h]
\centering
\includegraphics[width=\columnwidth]{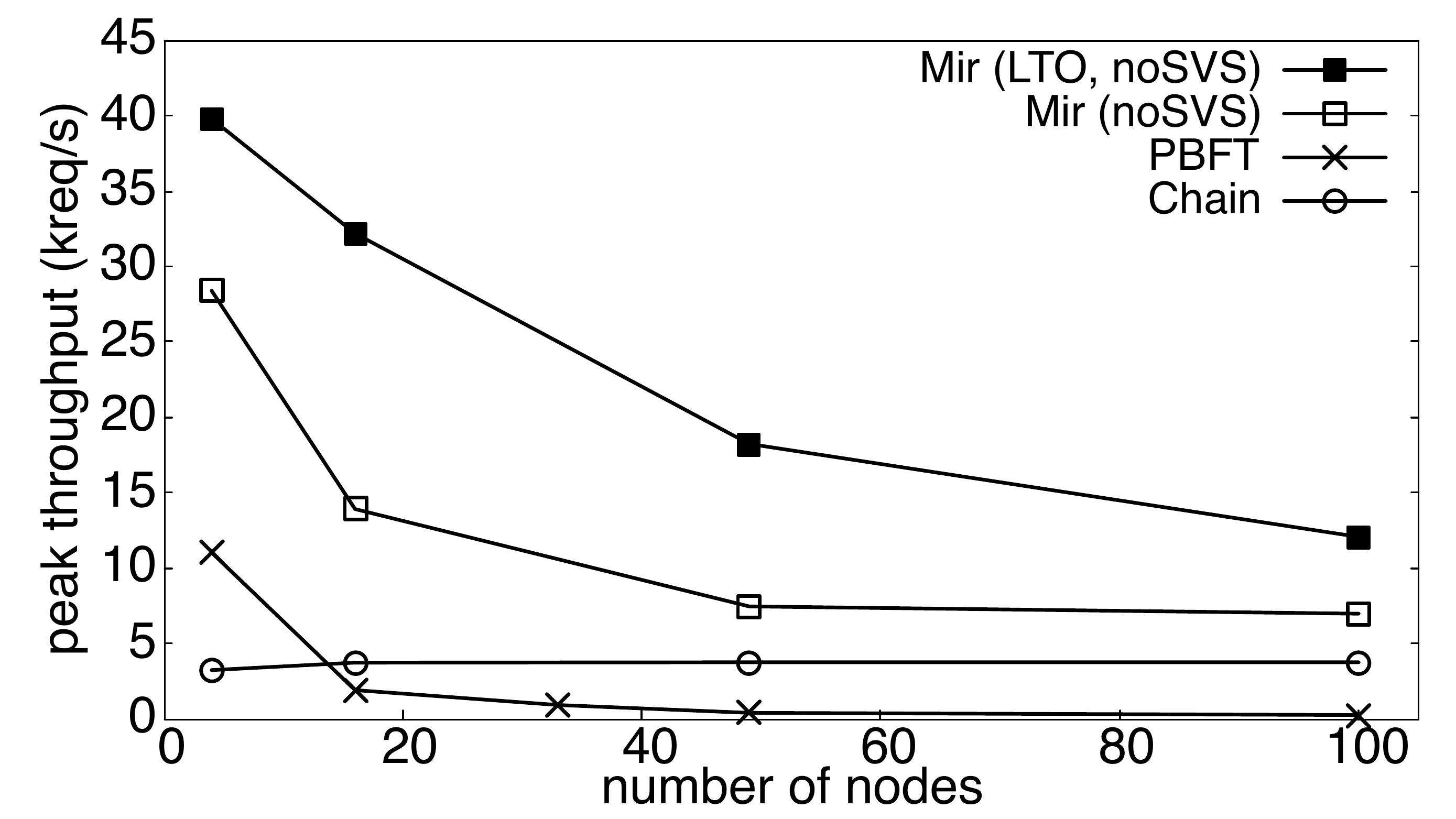}
\caption{WAN scalability experiment with large payload (3500 bytes).}
\label{fig:3500WAN}
\end{figure}

For large requests, where network bandwidth is the bottleneck, throughput of \protocol (with no SVS) reduces to 7k req/s with 100 nodes,
with a drop from 28.3k rew/s with $n=4$ nodes, see Figure~\ref{fig:3500WAN}.
We attribute this in part to the heterogeneity of VMs across datacenters (despite the identical specifications)
and, most importantly, to the non-uniform partition of the available uplink bandwidth.
Nevertheless, \protocol delivers the best performance of all protocols to date with 100 nodes on a WAN,
even compared to very optimistic protocols such as Chain, which delivers consistent throughput of about 4.5k req/s regardless of number of nodes. \protocol is, hence, the first robust BFT protocol which could be used as an ordering service in Fabric with $n=100$ nodes, without making ordering service a bottleneck (validation in Fabric is currently capped at less than 4k transactions per second \cite{AndroulakiBBCCC18}).

In addition to \protocol (with no SVS), Chain and PBFT, Figure~\ref{fig:3500WAN} also shows an experimental variant of \protocol which implements what we call \emph{Light Total Order (LTO)} broadcast, instead of full TOB (labeled `Mir (LTO, noSVS)'). LTO is an optimization, counterpart of SVS, to help alleviate network bottlenecks in TOB.  In short, LTO broadcast is identical to TOB,
except that it provides partial data availability
guaranteeing the delivery of the \emph{payload} of every request to \emph{at least one} correct node. This entails replicating batch payload to $f+1$ nodes in stable epoch, compared to all nodes without LTO.
Other correct nodes get and agree on the order of cryptographic hashes of requests,
which is the basis for maintaining other TOB properties. We provide more insight into LTO in Section~\ref{sec:LTO}.

LTO boosts throughput of \protocol to 40k 3500-byte req/s with $n=4$ nodes (roughly 40\% throughput improvement over \protocol) and maintains about 12.5k req/s with $n=100$ nodes (70\% throughput improvement over \protocol).

\subsection{Scalability in a Cluster/Datacenter}
\label{sec:LANeval}

\begin{figure}
	\centering
	\includegraphics[width=\columnwidth]{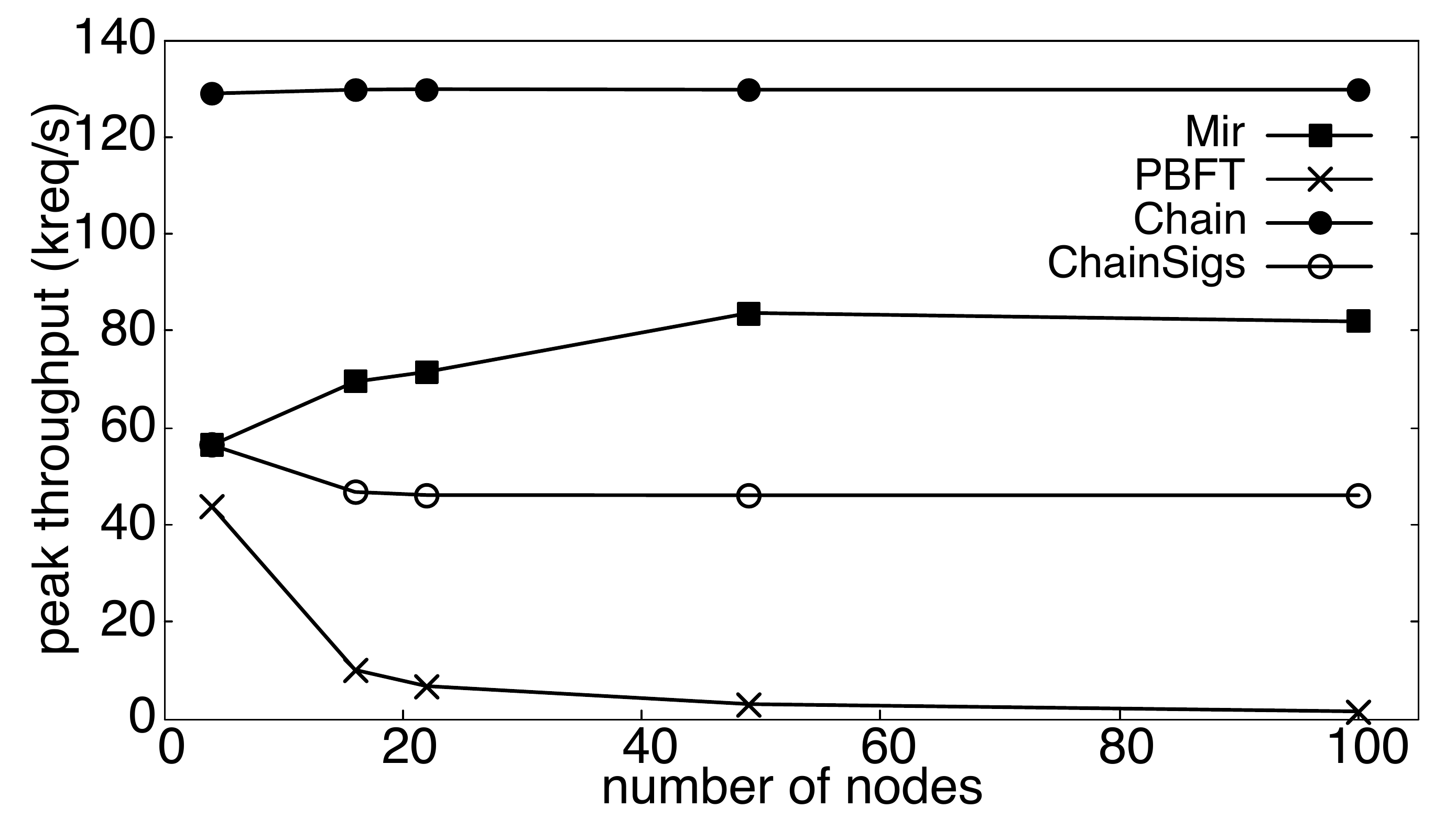}
	\caption{Throughput performance of Mir compared to Chain and PBFT in a single datacenter.}
	\label{fig:LAN}
\end{figure}

Figure~\ref{fig:LAN} depicts fault-free performance in a single datacenter with up to $n=100$ nodes.
Chain (the highest throughput BFT protocol to date in clusters) outperforms \protocol,
delivering roughly 1.6x of \protocol's peak throughput (130k req/s vs 83k req/s).
This difference is due to difference in client authentication: \protocol verifies clients' signatures, whereas
Chain uses vectors of MACs to authenticate a request to $\faults+1$ replicas (these are vulnerable to ``faulty client'' attacks \cite{Aardvark}).
Indeed, as soon as we add clients' signatures to Chain as in robust version of Chain \cite{700} (denoted by ChainSigs in Fig.~\ref{fig:LAN}),
Chain's throughput drops below that of \protocol.
\protocol maintains more than 80k req/s throughput,
significantly outperforming PBFT.

\subsection{Impact of optimizations and bucket rotation.}
\label{sec:OPTeval}
Fig.~\ref{fig:LT} shows the average latency and throughput
of different flavors of \protocol in fault-free executions using $n=16$ nodes.
We also show the performance of Chain and PBFT as a reference.
Nodes are distributed over $16$ distinct datacenters across the world.

\protocol without signature sharding (``Mir (noSVS)'') in Fig.~\ref{fig:LT})
saturates at roughly 53k req/s (resp. 12.3k req/s for large requests),
an approximate overhead of about 3\% (resp. 9.5\%)
compared to an idealized non-robust version of \protocol which involves no bucket rotation (``Mir (noRotation)''). Hence, the penalty of robust bucket rotation in \protocol is small. It is more than compensated for by signature sharding
which boosts \protocol throughput to 74k req/s (resp. to 33.5k req/s with LTO).

All variants of \protocol maintain roughly from 1--2s latency at relatively low load, to 3--5s latency close to saturation.  PBFT latency is lower at 600--800 ms, yet PBFT saturates under very low load compared to \protocol. We measured latency by: (1) synchronizing clocks between a client and a node belonging to the same datacenter with NTP, (2)  deducting request timestamp at a client from commit timestamp at a node, (3) averaging across all requests (and, consequently, all datacenters).

\begin{figure*}
    \begin{subfigure}{\columnwidth}
        \centering
        \includegraphics[width=\linewidth]{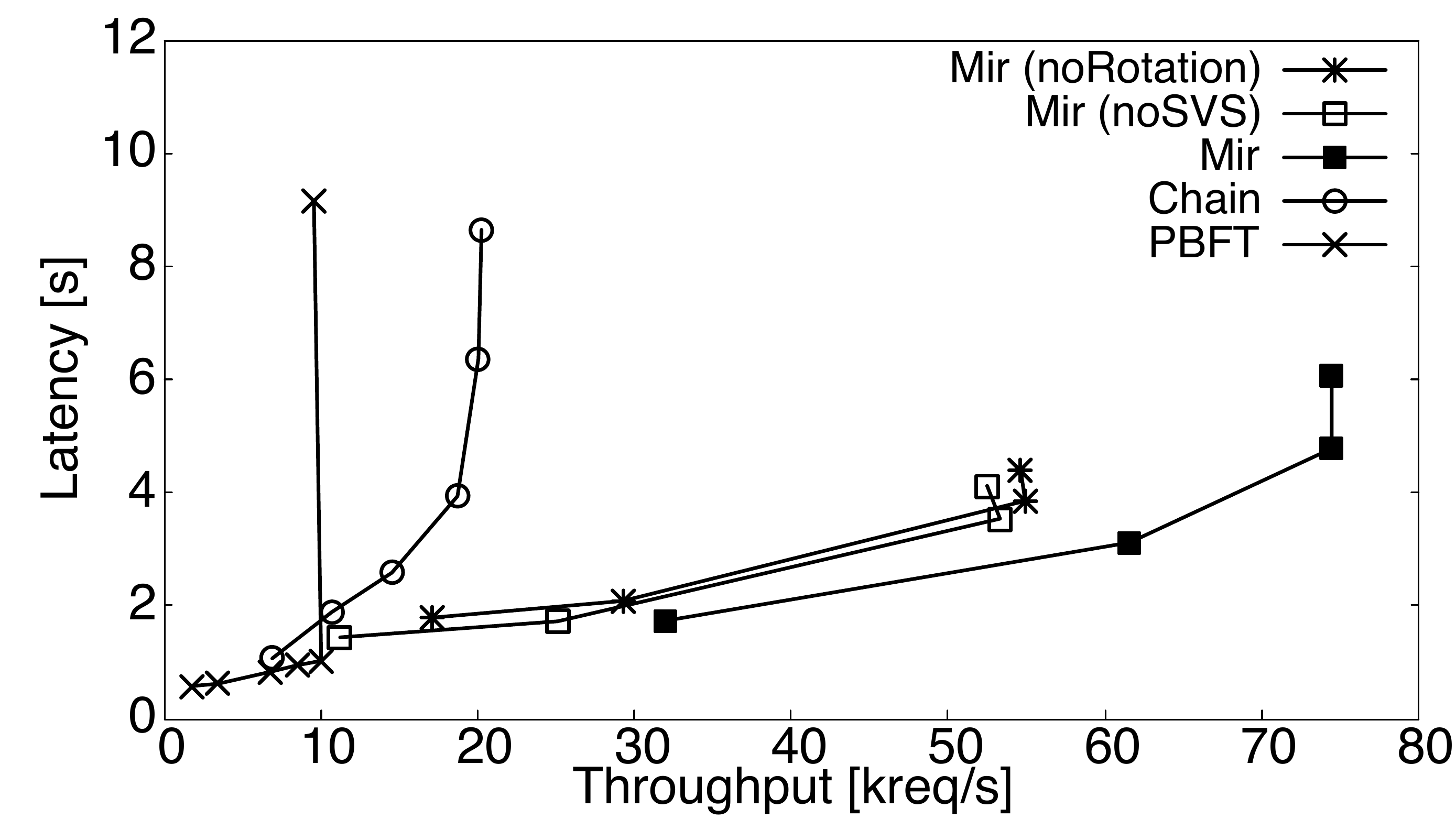}
        \caption{500 byte payload}
        \label{fig:500LT}
    \end{subfigure}%
    \begin{subfigure}{\columnwidth}
        \centering
        \includegraphics[width=\linewidth]{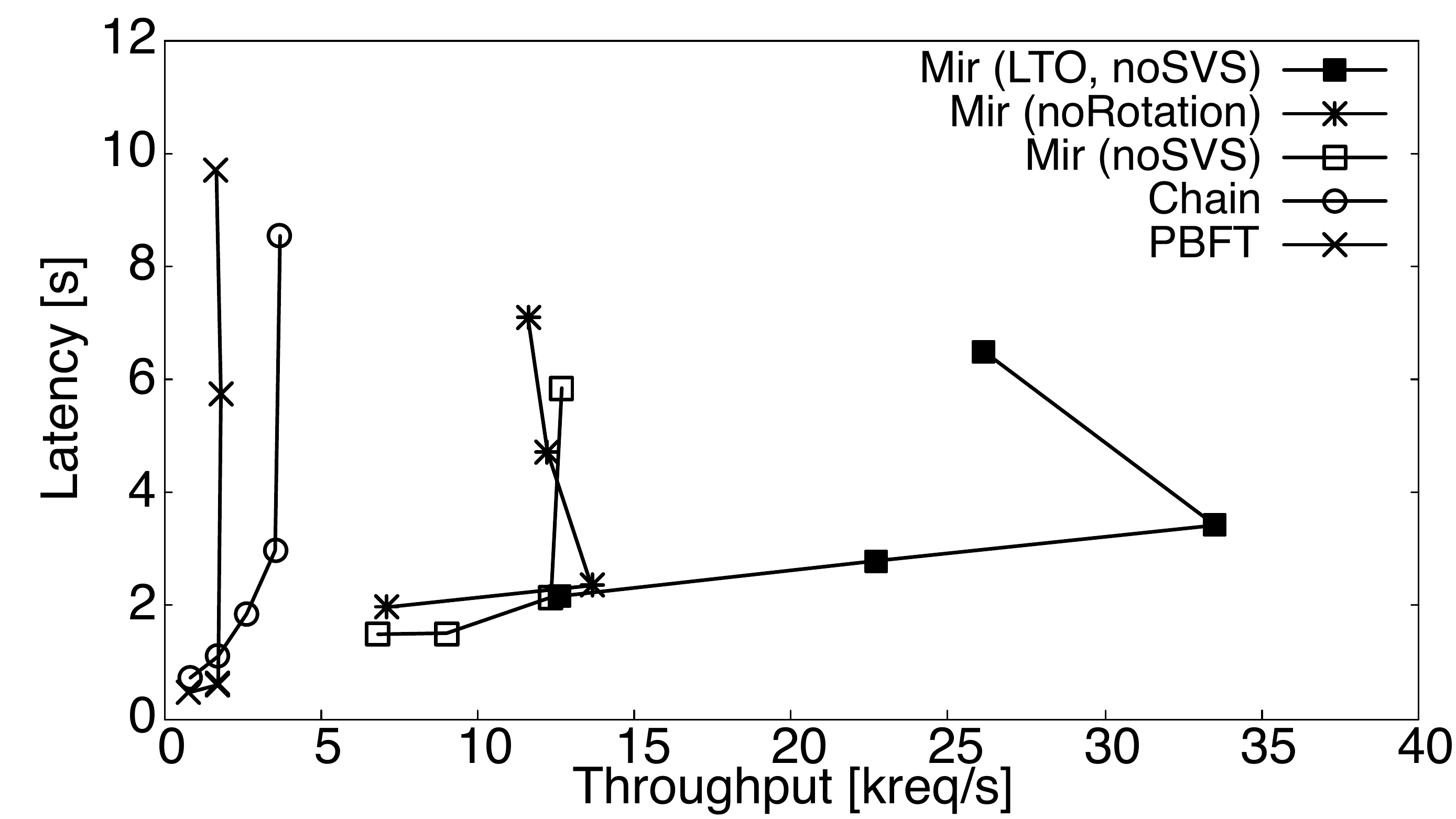}
        \caption{3500 byte payload}
        \label{fig:3500LT}
    \end{subfigure}
    \caption{Impact of bucket rotation and \protocol optimizations on a WAN with n=16 nodes.}
    \label{fig:LT}
\end{figure*}

\subsection{Benefits of Duplication Prevention}\label{sec:duplicates}
\begin{figure}
	\centering
	\includegraphics[width=\columnwidth]{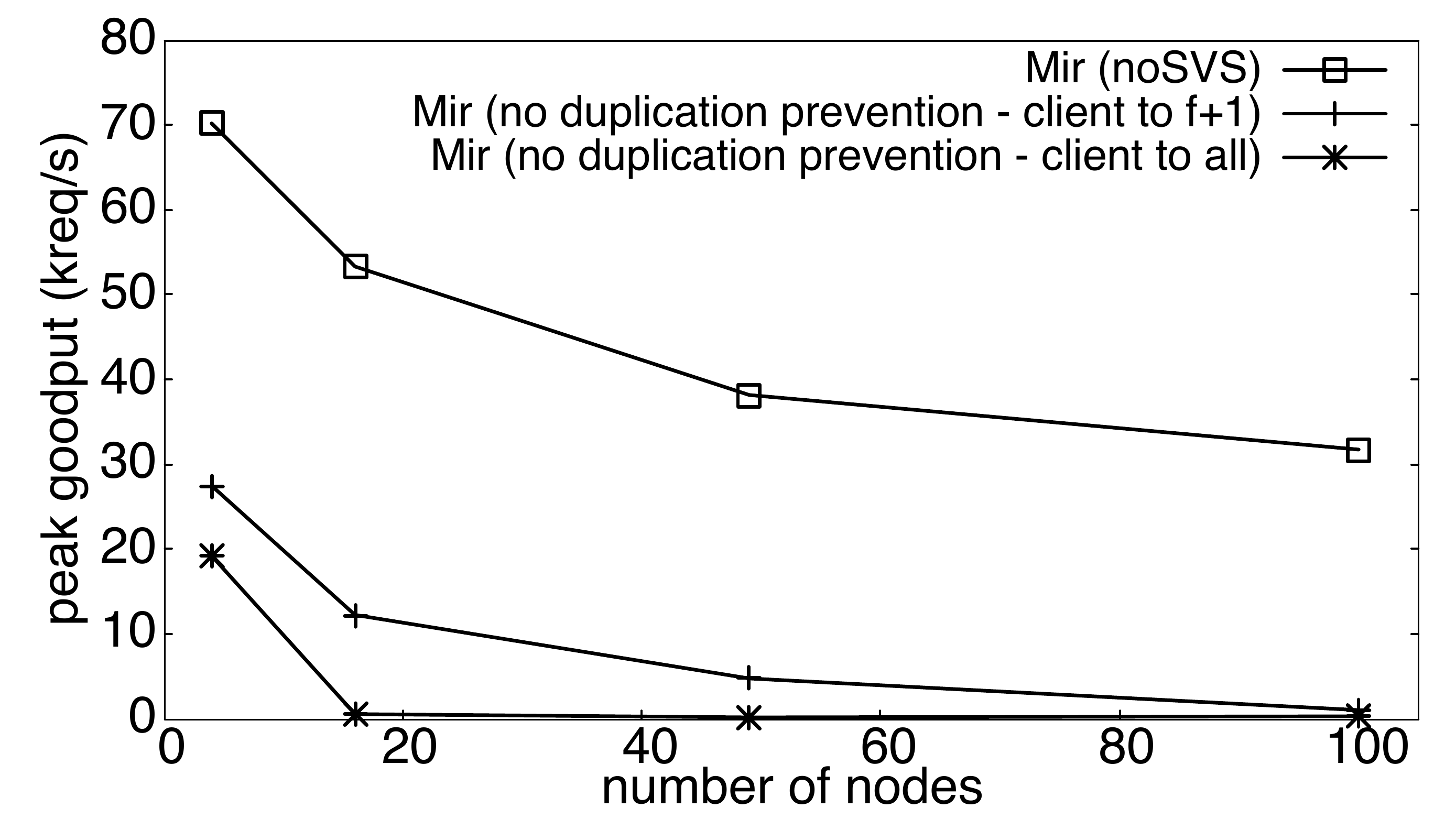}
	\caption{The impact of duplication prevention on a WAN with $n=16$.}
	\label{fig:duplicates500}
\end{figure}
In this section we examine the impact of duplicate requests to \emph{goodput},
i.e., throughput of unique requests.
In Fig.~\ref{fig:duplicates500} we compare the performance of Mir (noSVS)
to a version  of \protocol where the leaders do not partition requests in buckets,
but rather add in batches all their available requests, following what Hashgraph \cite{Hashgraph}, Red Belly \cite{RedBelly}, and BFT-Mencius \cite{BFT-Mencius} parallel leader protocols do.

We examine the impact of duplicates in two scenarios, (1)
where clients submit their requests to $f+1$ nodes --- intuitively, this is the minimum number of nodes to which a client must submit a request in any BFT protocol that insures liveness (due to possible censoring by f nodes), and (2) where clients submit their requests to all nodes.

The impact is a hefty performance penalty of 61\% (resp., 72\%) reduction in goodput compared to Mir (noSVS) in the first (resp,. second) scenario on $n=4$ nodes. This reaches as much as 97\% (resp., 99\%) with $n=100$, demonstrating $O(n^{-1})$ goodput scalability in protocols with duplication.

\subsection{Performance Under Faults}
\label{sec:evalfaults}
\begin{figure*}[t!]
    \centering
    \includegraphics[width=0.85\textwidth,clip=true,trim=0 10 12 10]{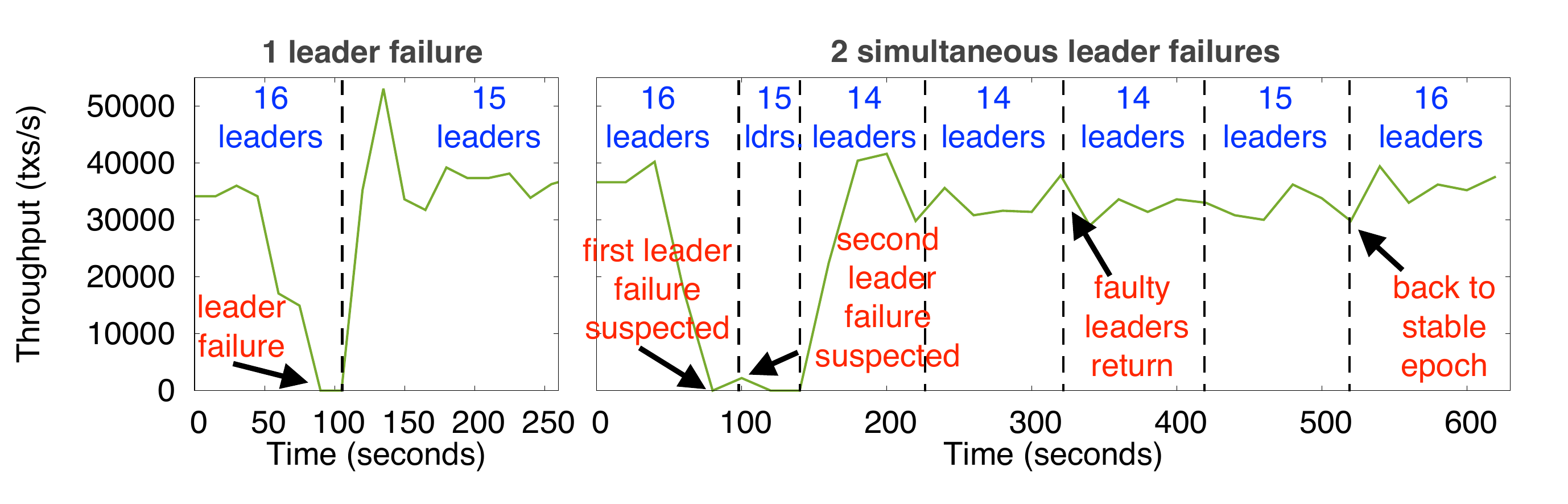}
    \caption{Performance under crash faults.}
    \label{fig:faults}
\end{figure*}
\begin{figure*}
	\begin{subfigure}{\columnwidth}
		\centering
		\includegraphics[width=\columnwidth]{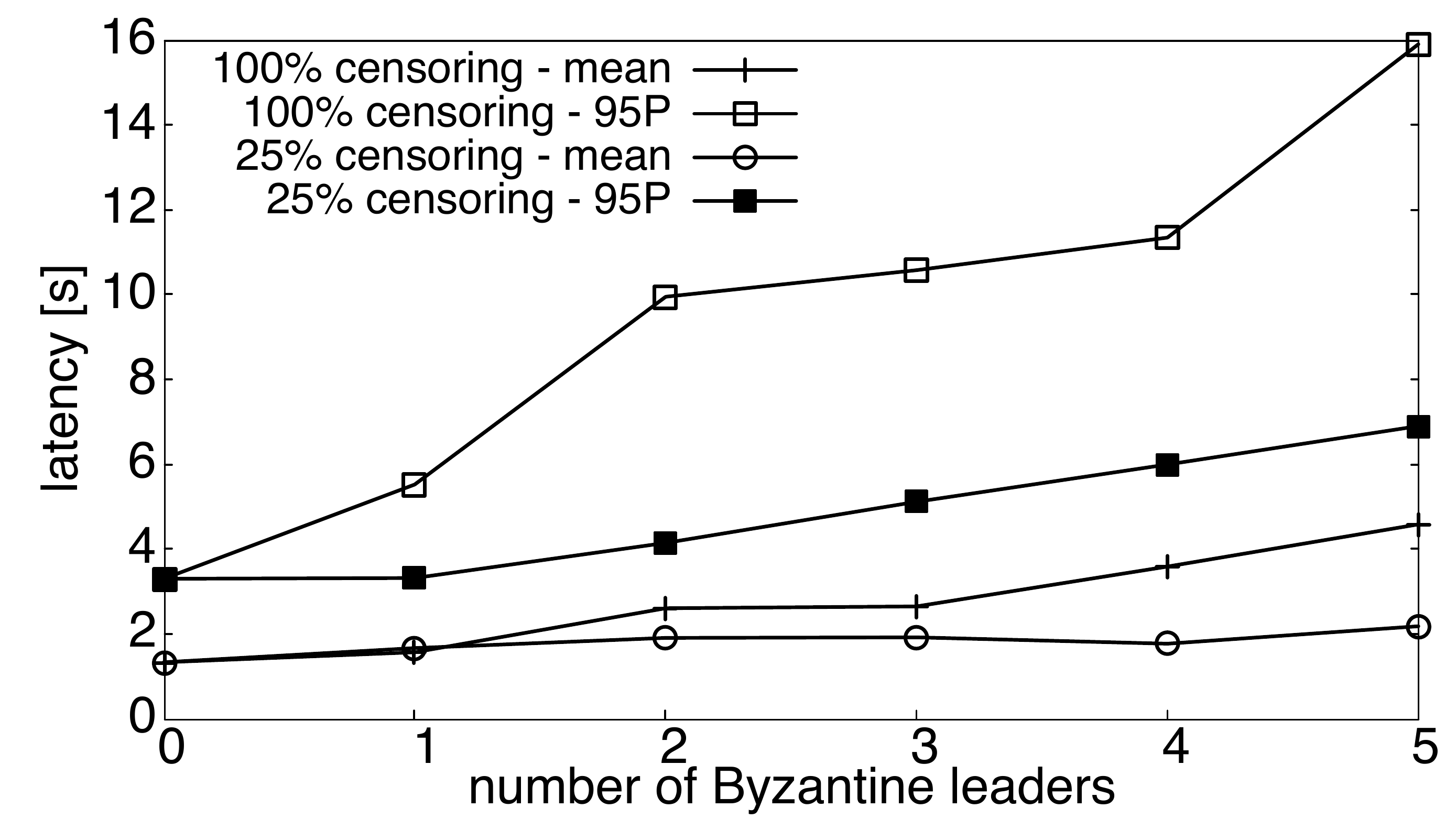}
		\caption{Mean and tail latencies (95\%) for increasing number of \\ Byzantine leaders that drop 25\% or 100\% of their requests.}
		\label{fig:censoring500}
	\end{subfigure}%
	\begin{subfigure}{\columnwidth}
		\centering
		\includegraphics[width=\linewidth]{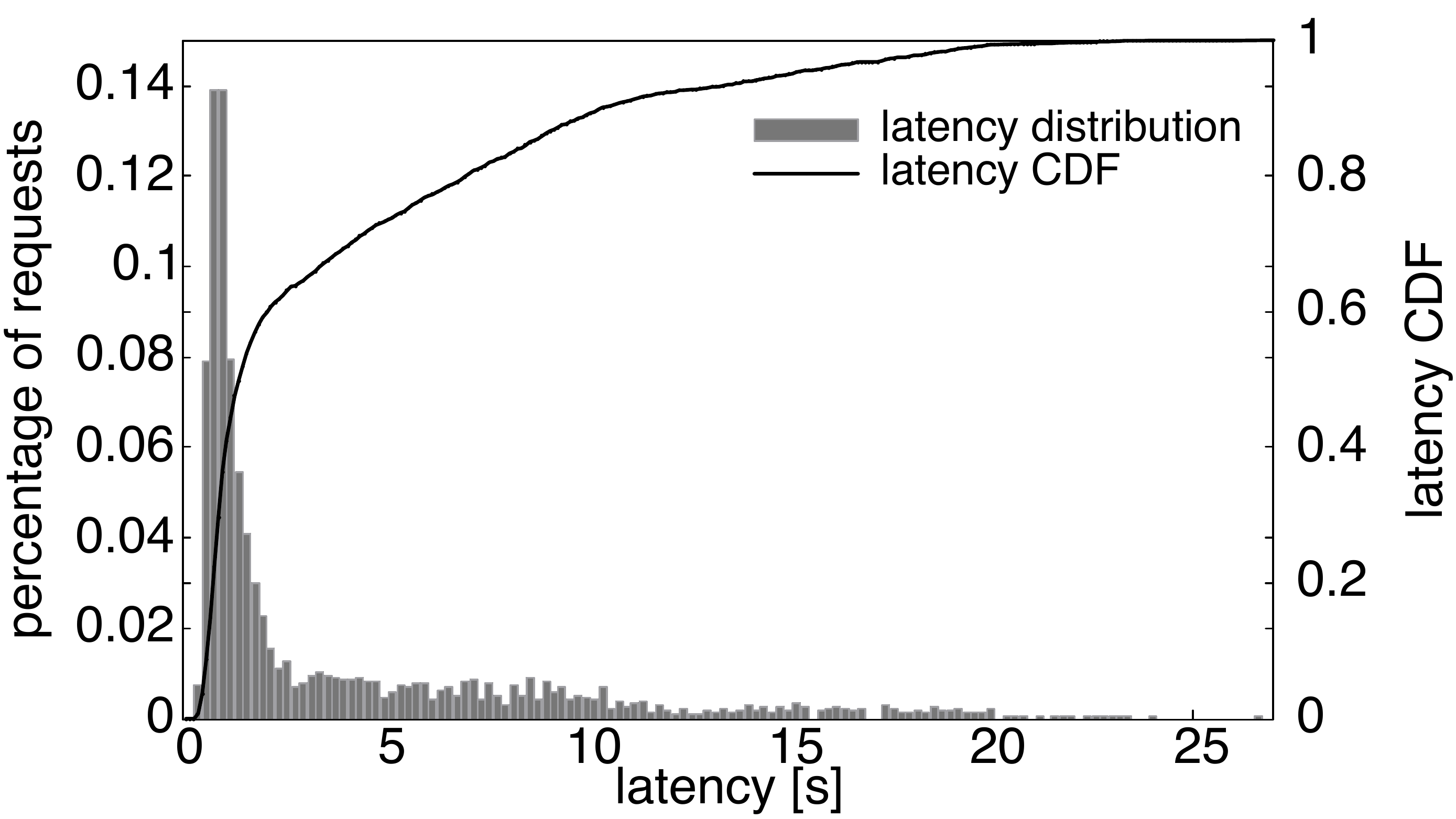}
		\caption{Latency distribution and CDF with $f=5$  Byzantine leaders censoring 100\% of requests.}
		\label{fig:cdf}
	\end{subfigure}
	\caption{Latency under the request censoring attack.}
	\label{fig:censoring}
\end{figure*}

\noindent\textbf{Leader Crash Faults.}
Figure~\ref{fig:faults} shows throughput as a function of time
when one and two leaders fail simultaneously.
We run this experiment in a WAN setting with 16 nodes,
and trigger a view change if an expected batch is not delivered with fixed timeouts of 20 seconds.
With one leader failure,
a view change is triggered and the system immediately transitions to a configuration with 15 leaders.
When two failures occur simultaneously,
the first view changes takes the system to a configuration with 15 leaders.
The first few batches are delivered in this configuration, but, since one of the 15 leaders has failed,
a second view change is triggered that takes the system to a configuration with 14 leaders,
from which execution can continue normally.
In this scenario, the figure also depicts the evolution of the leader set in case the failed nodes recover:
within three epochs, the system is in a stable state with 16 leaders again.

We can observe that gracious epoch changes are seamless in \protocol
(these occur from second 141 onwards in the experiment with 2 faults),
whereas ungracious epoch changes
(when throughput temporarily drops to $0$)
last approximately one epoch change timeout\remove{s}.

\noindent\textbf{Request Censoring (Byzantine Leaders Dropping Requests). }
In this experiment we emulate Byzantine behavior
by having an increasing number (from $0$ to $f=5$) of Byzantine leaders dropping (censoring) requests in our 16-node WAN setup.
Fig.~\ref{fig:censoring500} shows that mean latency remains below 4.6s (resp., 2.2s)
when Byzantine leader drops 100\% (resp. 25\%) of the requests they receive.
Tail latencies (95th percentile) remain below 16s (resp., 7s).
Fig.~\ref{fig:cdf} shows distribution and CDF of latency with $5$ Byzantine leaders censoring 100\% of requests.
When clients sent to all nodes we observe a drop up to 15\% for mean and 18\%
for tail latencies. A decrease of up to 44\% for mean and 49\% for tail latencies we
observe by reducing bucket rotation period to $128$ batches.
This introduces, though, a trade-off of approximately 10\% in peak throughput.

\noindent\textbf{Stragglers (Byzantine Leaders Delaying Proposals).}
In this experiment we evaluate \protocol resistance to stragglers.
Stragglers delay the batches they lead and propose empty batches.
The key to \protocol straggler resistance is that
a correct node starts an epoch change timeout for sequence number $sn$ as soon as it commits $sn-1$.
With multiple leaders proposing and committing batches independently,
a straggler can only impose a delay of one epoch change timeout \emph{once per epoch} without being detected,
as compared to once per sequence number in single-leader protocols.

We perform both WAN and LAN experiments with $n=16$ nodes, starting from a stable epoch.
The load is set at about 25-30\% peak throughput (corresponding to roughly 25k req/s).
Epoch change timeout is set to 20s and ephemeral epoch length to 256 batches.
We run our experiment until the straggler is removed from and re-added to the leader set (when it becomes epoch primary).

On WAN, fault-free throughput gives a baseline of 24.8k req/s.
With a single Byzantine straggler leader delaying each of its batches by 15s,
the average throughput is 18k req/s (penalty of 25\% over the baseline).
The straggler is always detected and removed from the leader set almost immediately.

On LAN, baseline throughput without faults is 28.1k req/s.
For reference, \protocol latency in LAN is in milliseconds.
We set straggler delay to 2 seconds (while keeping epoch change timeout to, for LAN very big, 20s) to keep the straggler longer in the leader set.
This time, straggler remains in the leader set for over 600 sequence numbers,
after which it is removed from the leader set.
In this case, we measure average throughput of 15.7k req/s in the entire execution (a penalty of 44\%).

To put these numbers into perspective,
a single-leader Aardvark \cite{Aardvark} suffers a 90\% performance penalty with a straggler primary on a LAN delaying batches for 10ms.
We conclude that \protocol has very good performance in presence of stragglers,
even with simple fixed epoch change timeouts.
Future optimizations of \protocol Byzantine node detection are possible,
following the approaches of Aardvark \cite{Aardvark} and RBFT \cite{RBFT}.


\section{Related Work}
\label{sec:relwork}

The seminal PBFT \cite{Castro:2002:PBF} protocol sparked intensive research on BFT. PBFT itself has a single-leader network bottleneck and does not scale well with the number of nodes.
\protocol generalizes PBFT and removes this bottleneck with a multi-leader approach,
enforcing a robust request duplication prevention.
Request duplication elimination is simple in PBFT and other single-leader protocols,
where this is the task of the leader.

Aardvark \cite{Aardvark} was the one the first BFT protocols, along with \cite{Spinning,Prime,RBFT},
to point out the importance of BFT protocol \emph{robustness},
i.e., guaranteed liveness and reasonable performance in presence of active denial of service and performance attacks.
In practice, Aardvark is a hardened PBFT protocol that uses clients' signatures, regular periodic view-changes (rotating primary),
and resource isolation using separate NICs for separating client\add{-to}-node from node\add{-to}-node traffic.
\protocol implements all of these and is thus robust in the Aardvark sense.
Beyond Aardvark features, \protocol is the first protocol to combine robustness with multiple leaders,
preventing request duplication performance attacks, enabling \protocol's excellent performance.

The first replication protocol to propose the use of multiple parallel leaders was Mencius \cite{Mencius}.
Mencius is a crash-tolerant Paxos-style \cite{Lamport:1998:PP:279227.279229} protocol
that leverages multiple leaders to reduce the latency of replication on WANs,
an approach  later followed by other crash-tolerant protocols (e.g., EPaxos \cite{Moraru:2013:MCE:2517349.2517350}).
The approach was extended to the BFT context by BFT-Mencius \cite{BFT-Mencius}.
Mencius and BFT-Mencius are geared towards optimizing latency and shard clients' requests by mapping a client to a closest node.
However, as a node can censor the request, a client is forced and allowed
to re-transmit the request to other nodes exposing a vulnerability to request duplication attacks which BFT-Mencius does not handle.
As illustrated in our evaluation (Sec.~\ref{sec:duplicates}), malicious clients can severely impact the throughput of such a scheme, by sending their requests to multiple or all nodes. Unlike in a regular DoS attack, these clients cannot be naively declared Byzantine or rate-limited, as such request traffic may be needed by correct clients to deal with Byzantine leaders dropping requests (request censoring attack) or to optimize the latency of a BFT protocol.  Unlike BFT-Mencius,  \protocol maps clients' requests to buckets
which are then assigned to nodes, similarly to consistent hashing \cite{Karger:1997:CHR:258533.258660}.
\protocol further redistributed bucket assignment in time to enforce robustness to request censoring.
Unlike Mencius, EPaxos and BFT-Mencius, \protocol does not optimize for latency in the best case,
paying a small price as it does not assign clients to the closest nodes. However, our experiments show that this impact is acceptable,
in particular given that the blockchain is not the most latency-sensitive application.

Recent BFT protocols, proposed in the blockchain context \cite{RedBelly,Hashgraph},
that exhibit multi-leader flavor, also do not address request duplication.
Furthermore, unlike \protocol, these proposals invent new BFT protocols from scratch which is a highly error-prone and tedious process \cite{700}.
In contrast, \protocol follows an evolutionary rather then revolutionary design approach to a multi-leader protocol,
building upon proven PBFT/Aardvark algorithmic and systems' constructs,
considerably simplifying the reasoning about \protocol correctness.

Two recent protocols, HotStuff~\cite{hotstuff} and SBFT~\cite{SBFT},
are leader-based protocols that improve on PBFT's quadratic common-case message complexity
and require a linear ($O(n)$) number of messages in the common case.
HotStuff is optimized for throughput and features $O(n)$ messages in view change as well
(SBFT requires $O(n^2)$ messages in view change).
While \protocol approach of multiplexing PBFT instances and SBFT/HotStuff improvements over PBFT appear largely orthogonal,
our experiments show that \protocol multi-leader approach scales better than HotStuff, which is a single-leader protocol.
Namely, even though PBFT/\protocol have quadratic common-case message complexity, these messages are load balanced across $n$ nodes,
yielding $O(n)$ messages at a \emph{bottleneck replica}, just like HotStuff/SBFT.
Our experiments also showed that HotStuff retains the
downside of other single-leader protocols,
i.e., bottlenecks related to leader sending all proposals, yielding an infavorable $O(n^{-1})$ throughput scalability trend.
An unimplemented HotStuff variant, called ChainedHotStuff \cite{hotstuff},  suggests having different leaders piggyback their batches on other protocol common-case messages.
As Hotstuff has 4 common case phases,
this allows up to 4 ``chained'' leaders in ChainedHotStuff regardless of the total number of nodes,
which is less efficient than \protocol which allows up to $n$ parallel leaders.
In future, it would be very interesting to combine the two approaches, $O(n)$ common case message complexity and parallel leaders, by implementing \protocol variants based on HotStuff/SBFT instead of PBFT.

\emph{Optimistic} BFT protocols \cite{700,Zyzzyva} have been shown to be very efficient on a small scale in clusters.
In particular, Aliph \cite{700} is a combination of Chain crash-tolerant replication \cite{Chain}
ported to BFT and backed by PBFT/Aardvark outside the optimistic case where all nodes are correct.
We demonstrated that \protocol holds its ground with BFT Chain in clusters and it considerably outperforms it in WANs.
Nevertheless, \protocol remains compatible with the modular approach to building optimistic BFT protocols of \cite{700},
where \protocol can be used as a robust and high-performance backup protocol.
Zyzzyva \cite{Zyzzyva} is an optimistic leader-based protocol that optimizes for latency.
While we chose to implement \protocol based on PBFT,
\protocol variants based on Zyzzyva's latency-efficient communication pattern are conceivable.

Eventually synchronous BFT protocols, to which \protocol belongs,
circumvent the FLP consensus impossibility result \cite{FLP} by assuming eventual synchrony.
These protocols, \protocol included, guarantee safety despite asynchrony but rely on eventual synchrony to provide liveness.
Alternatively, probabilistic BFT protocols such as Honeybadger \cite{MillerXCSS16} and BEAT \cite{BEAT}
provide both safety and liveness (except with negligible probability) in purely asynchronous networks.
By comparing Honeybadger and \protocol, we showed that this comes as a tradeoff,
as \protocol significantly outperforms Honeybadger, even though both protocols target the same deployment setting (up to 100 nodes in a WAN).
Notably, Honeybadger authors realize the importance of duplicate elimination
and suggest that each leader randomly samples the requests in their pending queue.
This approach would result to no duplicates on expectation.
However, in practice, unless the system is deep in saturation,
the pending queue does not contain significantly more requests than the next batch.
Indeed, in our Honeybadger evaluation we observed that goodput (effective throughput) was roughly only $20\%$ of the nominal throughout.
BEAT suggests some optimizations over Honeybadger without significantly outperforming the former.

As blockchains brought an arms-race to BFT protocol scalability \cite{Vukolic15},
many proposals focus on large, Bitcoin-like scale, with thousands or tens of thousands of nodes \cite{Algorand,Bitcoin-NG}.
In particular, Algorand \cite{Algorand} is a recent BFT protocol that deals with BFT agreement in populations of thousands of nodes,
by relying on a verifiable random function to select a committee in the order of hundred(s) of nodes.
Algorand then runs a smaller scale agreement protocol inside a committee.
We foresee \protocol being a candidate for this ``in-committee'' protocol inside a system such as Algorand
as well as in other blockchain systems that effectively restrict voting to a smaller group of nodes,
as is the case in Proof of Stake proposals \cite{CasperFFG}.
In addition, \protocol is particularly interesting to permissioned blockchains, such as Hyperledger Fabric \cite{AndroulakiBBCCC18}.

ByzCoin~\cite{ByzCoin} scales PBFT for permissionless blockchains by building PBFT atop of CoSi \cite{CoSi}, a collective signing protocol that efficiently aggregates hundreds or thousands of signatures. Moreover, it adopts ideas from PoW based Bitcoin-NG \cite{Bitcoin-NG} to decouple transaction verification from block mining. This approach is orthogonal to that of Mir and variants of Byzcoin with Mir instead of PBFT are interesting for future work.

Stellar\cite{stellarsosp} uses SCP, a Byzantine agreement protocol
with asymmetric quorums and trust assumptions targeting payment networks, which targets similar network sizes  with \protocol.
 Asymmetric quorums of SCP modify trust assumptions and the liveness guarantees of traditional BFT protocols, with \cite{stellarattack} showing liveness violation with failures of only two specific nodes in a production configuration of Stellar.
We show it is possible to obtain high throughput and low latencies
while maintaining the strong guarantees of BFT protocols with classical (symmetric) quorums and trust assumptions.

Finally, \emph{sharding} protocols \cite{elastico,omniledger} partition transaction verification into independent shards.
\protocol is complementary to such protocols as they either require ordering within a shard or total ordering of the shards.
Monoxide \cite{monoxide} also uses sharding to increase throughput,
but provides weaker guarantees (eventual atomicity across shards).
Moreover, Monoxide's scalability heavily depends on transaction payload semantics.


\section{Conclusions}
\label{sec:conclusions}

This paper presented \protocol, a high-throughput robust BFT protocol for decentralized networks.
\protocol is the first BFT protocol that uses multiple parallel leaders
thwarting both censoring attacks and request duplication performance attacks.
In combination with reducing CPU overhead through the ``signature verification sharding'' optimization,
this allows \protocol to achieve unprecedented throughput at scale even on a wide area network,
outperforming state-of-the-art protocols.

The main insight behind \protocol is multiplexing multiple parallel instances of the PBFT protocol into a single
totally ordered log, while preventing duplicate request proposals by partitioning the request hash space
and assigning each subset to a different leader.
\protocol prevents request censoring attacks by periodically changing this assignment to
guarantee that each request is eventually assigned to a correct leader.
Being based on well understood and thoroughly scrutinized PBFT
makes it is easy to reason about \protocol's correctness.


\section*{Acknowledgements}
We thank Jason Yellick for his very insighful comments and suggestions.

This work was supported in part by the European Union's Horizon 2020 Framework Programme under gtrant aggreement number 780477 (PRIViLEDGE).


\balance

\end{document}